\documentclass[11pt]{article} %%%%%%%%%%%%%%%%%%%%%%%%%%%%%%%%%%%%%%%%%%%%%%%%%

\usepackage{textcomp, amssymb}
\usepackage{anysize} 
\usepackage{amsfonts}
\usepackage{amsmath}
\usepackage{dsfont}
\usepackage{MnSymbol}
\usepackage{mathtools}
\usepackage{graphicx}
\usepackage{epsfig}
\usepackage{epstopdf}
\usepackage{tikz}
\usepackage{amsthm}
\usepackage{ifpdf}
\usepackage{thm-restate}
\usepackage{stmaryrd}
\usepackage{cleveref}
\usepackage{subcaption}
\usepackage{longtable}

\usetikzlibrary{calc,decorations.pathmorphing,shapes,graphs,positioning,automata,angles}

\newcounter{sarrow}

\newcounter{sarrow1}
\newcommand\xnrsquigarrow[1]{%
\stepcounter{sarrow1}%
\mathrel{\begin{tikzpicture}[baseline= {( $ (current bounding box.south) + (0,-0.5ex) $ )}]
\node[inner sep=.5ex] (\thesarrow) {$\scriptstyle #1$};
\path[draw,<-,decorate,
  decoration={zigzag,amplitude=0.7pt,segment length=1.2mm,pre=lineto,pre length=4pt}]
    (\thesarrow1.south east) -- (\thesarrow1.south west);
    $\slashedarrowfill@\relbar\relbar/$
    \end{tikzpicture}}%
}

\makeatletter
\def\slashedarrowfill@#1#2#3#4#5{%
  $\m@th\thickmuskip0mu\medmuskip\thickmuskip\thinmuskip\thickmuskip
   \relax#5#1\mkern-7mu%
   \cleaders\hbox{$#5\mkern-2mu#2\mkern-2mu$}\hfill
   \mathclap{#3}\mathclap{#2}%
   \cleaders\hbox{$#5\mkern-2mu#2\mkern-2mu$}\hfill
   \mkern-7mu#4$%
}
\def\rightslashedarrowfillb@{%
  \slashedarrowfill@\relbar\relbar/\rightarrow}
\newcommand\xnrightarrow[2][]{%
  \ext@arrow 0055{\rightslashedarrowfillb@}{#1}{#2}}

\def\rightslashedarrowfille@{%
  \slashedarrowfill@\relbar\relbar/\twoheadrightarrow}
\newcommand\xntworightarrow[2][]{%
  \ext@arrow 0055{\rightslashedarrowfille@}{#1}{#2}}

\def\rightslashedarrowfillg@{%
  \slashedarrowfill@\relbar\relbar{\raisebox{.12em}{}}\twoheadrightarrow}
\newcommand\xtworightarrow[2][]{%
  \ext@arrow 0055{\rightslashedarrowfillg@}{#1}{#2}}

\def\rightslashedarrowfillx@{%
  \slashedarrowfill@\Relbar\Relbar/\rightrightarrows}
\newcommand\xnTworightarrow[2][]{%
  \ext@arrow 0055{\rightslashedarrowfillx@}{#1}{#2}}

\def\rightslashedarrowfilly@{%
  \slashedarrowfill@\Relbar\Relbar{\raisebox{.12em}{}}\rightrightarrows}
\newcommand\xTworightarrow[2][]{%
  \ext@arrow 0055{\rightslashedarrowfilly@}{#1}{#2}}

\pgfdeclareshape{slash underlined}
{
  \inheritsavedanchors[from=rectangle] % this is nearly a circle
  \inheritanchorborder[from=rectangle]
  \inheritanchor[from=rectangle]{north}
  \inheritanchor[from=rectangle]{north west}
  \inheritanchor[from=rectangle]{north east}
  \inheritanchor[from=rectangle]{center}
  \inheritanchor[from=rectangle]{west}
  \inheritanchor[from=rectangle]{east}
  \inheritanchor[from=rectangle]{mid}
  \inheritanchor[from=rectangle]{mid west}
  \inheritanchor[from=rectangle]{mid east}
  \inheritanchor[from=rectangle]{base}
  \inheritanchor[from=rectangle]{base west}
  \inheritanchor[from=rectangle]{base east}
  \inheritanchor[from=rectangle]{south}
  \inheritanchor[from=rectangle]{south west}
  \inheritanchor[from=rectangle]{south east}
  \inheritanchorborder[from=rectangle]
  \foregroundpath{
    \southwest \pgf@xa=\pgf@x \pgf@ya=\pgf@y
    \northeast \pgf@xb=\pgf@x \pgf@yb=\pgf@y
    \pgf@xc=\pgf@xa
    \advance\pgf@xc by .5\pgf@xb
    \pgf@yc=\pgf@ya 
    \advance\pgf@xc by -1.3pt
    \advance\pgf@yc by -1.8pt
    \pgfpathmoveto{\pgfqpoint{\pgf@xc}{\pgf@yc}}
    \advance\pgf@xc by  2.6pt
    \advance\pgf@yc by  3.6pt
    \pgfpathlineto{\pgfqpoint{\pgf@xc}{\pgf@yc}}
    \pgfpathmoveto{\pgfqpoint{\pgf@xa}{\pgf@ya}}
    \pgfpathlineto{\pgfqpoint{\pgf@xb}{\pgf@ya}}
 }
}
\tikzset{nomorepostaction/.code=\let\tikz@postactions\pgfutil@empty}
\newcommand{\runhead}[1]{\par\medskip\noindent\textbf{#1}}

\newcommand*{\rmbrace}{|\mskip-4mu\}}
\newcommand*{\lmbrace}{\{\mskip-4mu|}
\newcommand*{\mset}[1]{\lmbrace#1\rmbrace}
\newcommand{\pb}{\underline{\leftrightarrow}}
\newcommand{\pbb}{\underline{\leftrightarrow}_{b}}
\newcommand{\pbrb}{\underline{\leftrightarrow}_{rb}}
\newtheorem{theorem}{Theorem}[section]
\newtheorem{definition}[theorem]{Definition}

\newtheorem{lemma}[theorem]{Lemma}

\begin{document}
%%% tile page %%%%%%%%%%%%%%%%%%%%%%%%%%%%%%%%%%%%%%%%%%%%%%%%%%%%%%%%%%%%%%%%%

\begin{titlepage}
\thispagestyle{empty}

\hrule
\begin{center}
{\bf\LARGE An Algebra of Transformers\\}
%\vspace{0.3cm}
%
%{\bf --- Parallelism as a Fundamental Nature\\}
\vspace{0.5cm}
--- Yong Wang, Zhiyan Wang, Hexiao Kang ---

\vspace{2cm}
%\begin{figure}[!htbp]
% \centering
% \includegraphics[width=1.0\textwidth]{images/cover_medium.eps}
%\end{figure}

\end{center}
\end{titlepage}

\newpage %%%%%%%%%%%%%%%%%%%%%%%%%%%%%%%%%%%%%%%%%%%%%%%%%%%%%%%%%%%%%%%%%%%%%%

\setcounter{page}{1}\pagenumbering{roman}

\tableofcontents

\newpage
\setcounter{page}{1}\pagenumbering{arabic}

        \section{Introduction}\label{intro} 

Because of the parallelizability of transformers \cite{AIAYN}, transformers are becoming the mainstream techniques in large language models (LLMs). On the theoretic analyses on transforms, they mainly focus on the computability analysis, such the Turing-completeness of unbounded precision transformers \cite{AIC}, the advantage of parallelizability of unbounded precision transformers \cite{PT1} and the limitations of bounded precision transformers \cite{PT2}.

This paper give transformers an algebra for reasoning, which relies heavily on the functional decomposition results of transformers \cite{TTCB}, we follow its notations and definitions of transformers.

We write $[n]$ for the set $\{1,\cdots,n\}$, $X^*$ for the set of sequences over set $X$, $X^+$ for the set of non-empty sequences over $X$, string $\Sigma^*$ as $w=w_1\cdots w_n$ for a finite alphabet $\Sigma$, $\mathbf{0}$ for zero vector or matrix, $\mathbf{I}$ for identity matrix, $\mathbf{1}$ for a vector or matrix whose entries are all $1$'s. For any function $f:\mathbb{R}\rightarrow\mathbb{R}$ and any vector $\mathbf{x}=[x_i]_{i=1}^{d}\in\mathbb{R}^{d}$, $f$ is extended coordinate-wise: \[f(\mathbf{x})=f\left(\left[\begin{array}{c}
     x_1\\
     x_2\\
     \vdots\\
     x_d
  \end{array}\right]\right)=\left[\begin{array}{c}
     f(x_1)\\
     f(x_2)\\
     \vdots\\
     f(x_d)
  \end{array}\right]\]
  
We write $f:X^{+}\xrightarrow{lp}Y^{+}$ for a function from $X^{+}$ to $Y^{+}$, which is length-preserving. And we write $f.\mathsf{a}$ for a parameter $a$ of $f$.

\begin{definition}[Transformers]
A transformer is a function $\mathsf{tf}:\Sigma^{+}\xrightarrow{lp}\mathcal{Y}^{+}$, $\mathsf{tf}(x_1\cdots x_n)=y_1\cdots y_n$, where $\mathbf{z}^{(0)}_{i}=\mathsf{tf}.\mathsf{we}(x_{i})+\mathsf{tf}.\mathsf{pe}_{n}(i)$, $(\mathbf{z}^{(1)}_{1},\cdots,\mathbf{z}^{(1)}_{n})=\mathsf{tf}.\mathsf{tl}^{(1)}(\mathbf{z}^{(0)}_{1},\cdots,\mathbf{z}^{(0)}_{n})$, $\cdots$, $(\mathbf{z}^{(L)}_{1},\cdots,\mathbf{z}^{(L)}_{n})=\mathsf{tf}.\mathsf{tl}^{(L)}(\mathbf{z}^{(L-1)}_{1},\cdots,\mathbf{z}^{(L-1)}_{n})$, $\mathbf{y}_{i}=\mathsf{tf}.\mathsf{out}(\mathbf{z}^{(L)}_{i})$. Where $\mathsf{tf}.\mathsf{we}$ and $\mathsf{tf}.\mathsf{pe}_{n}$ are word embedding function and position embedding function, $\mathsf{tf}.\mathsf{tl}^{(\ell)}$ are transformer layers, $\mathsf{tf}.\mathsf{out}$ is an unembedding function with some output space $\mathcal{Y}$. For the details of these functions, please refer to \cite{TTCB}.
\end{definition}

The decomposition of transformers coming from \cite{TTCB} is shown in \cref{DoT}. The idea of decomposition of a transformer is dividing a transformer into a set of atomic constructions, which are operations of vectors or matrices. For the details of these constructions, please refer to \cite{TTCB}.

\begin{center}
    \begin{longtable}{|c|c|c|}
      \caption{Decomposition of Transformers}\\
      \hline Names &Functions& Types\\
      \hline
      \endfirsthead
      \multicolumn{3}{c}{\bfseries continuing: Decomposition of Transformers} \\
      \hline Names &Functions& Types\\
      \hline
      \endhead
      \hline
      \multicolumn{3}{r}{to be continued\ldots} \\
      \endfoot
      \hline
      \multicolumn{3}{r}{end of Decomposition of Transformers} \\
      \endlastfoot
      $\alpha_{1}$ & word embedding & $\Sigma\rightarrow \mathbb{R}^{d}$\\
      $\alpha_{2}$ & position embedding & $\mathbb{N}\rightarrow\mathbb{R}^{d}$\\
      $\alpha_{3}$ & identity & $\mathbb{R}^{d}\rightarrow\mathbb{R}^{d}$\\
      $\alpha_{4}$ & min/max & $\mathbb{R}^{2}\rightarrow\mathbb{R}$\\
      $\alpha_{5}$ & add/subtract & $\mathbb{R}^{2}\rightarrow\mathbb{R}$\\
      $\alpha_{6}$ & multiply by $c$ & $\mathbb{R}\rightarrow\mathbb{R}$\\
      $\alpha_{7}$ & multiply $xy$ & $\mathbb{R}^{2}\rightarrow\mathbb{R}$\\
      $\alpha_{8}$ & comparators & $\mathbb{R}\rightarrow\{0,1\}$\\
      $\alpha_{9}$ & boolean functions & $\{0,1\}^{m}\rightarrow\{0,1\}$\\
      $\alpha_{10}$ & continuous piecewise linear functions & $\mathbb{R}^{d}\rightarrow\mathbb{R}$\\
      $\alpha_{11}$ & cancel residual & $\mathbb{R}^{d}\rightarrow\mathbb{R}^{d}$\\
      $\alpha_{12}$ & identity & $(\mathbb{R}^{d})^{+}\rightarrow(\mathbb{R}^{d})^{+}$\\
      $\alpha_{13}$ & uniform average & $(\mathbb{R}^{d})^{+}\rightarrow(\mathbb{R}^{d})^{+}$\\
      $\alpha_{14}$ & first & $(\mathbb{R}^{d})^{+}\rightarrow(\mathbb{R}^{d})^{+}$\\
      $\alpha_{15}$ & index lookup & $(\mathbb{R}^{d})^{+}\rightarrow\mathbb{R}^{+}$\\
      $\alpha_{16}$ & predecessor & $(\mathbb{R}^{d})^{+}\rightarrow(\mathbb{R}^{d})^{+}$\\
      $\alpha_{17}$ & tie-breaking & $--$\\
      $\alpha_{18}$ & multi-head & $--$\\
      $\alpha_{19}$ & unembedding & $--$\\
      $\alpha_{20}$ & layer normalization & $\mathbb{R}^{d}\rightarrow\mathbb{R}^{d}$\\
      $\phi_{1}$ & conditional & $\{0,1\}\times\mathbb{R}^{2}\rightarrow\mathbb{R}$
      \label{DoT}
    \end{longtable}
\end{center}

Note that, the atomic constructions include two kinds: one for atomic actions, the other for guards in programming language. The set of atomic actions is denoted $\mathsf{Act}=\{\alpha_{1}, \alpha_{2}, \alpha_{3}, \alpha_{4}, \alpha_{5}\}\cup\{\alpha_{6}, \alpha_{7},\alpha_{8}, \alpha_{9}, \alpha_{10}, \alpha_{11}, \alpha_{12}, \alpha_{13}, \alpha_{14}, \alpha_{15}, \alpha_{16}, \alpha_{17}, \alpha_{18}, \alpha_{19}, \alpha_{20}\}$ and the set of guards is denoted $\mathsf{Grd}=\{\phi_{1}\}$. In the boolean functions, the truth value $\mathbf{true}$ is interpreted as number $1$ and $\mathbf{false}$ as $0$; in the general vectors or matrices operations, $\mathbf{true}$ is interpreted as $\mathbf{I}$ and $\mathbf{false}$ as $\mathbf{0}$. 
\newpage\section{Preliminaries}\label{pre} 

For self-satisfactory, in this chapter, we introduce the preliminaries on multiset in \cref{ms}, equational logic in \cref{el} and pomset labelled transition systems in \cref{plts}.

\subsection{Multiset}\label{ms}

\begin{definition}[Labelled poset]
A labelled poset is a tuple $\mathbf{u}=\langle S, \leq, \lambda\rangle$, where $S$ is the carrier set, $\leq$ is a partial order on $S$ and $\lambda$ is a labelling function $\lambda:S\rightarrow\Sigma$.

For a labelled poset $\mathbf{u}$, $S_{\mathbf{u}}$, $\leq_{\mathbf{u}}$ and $\lambda_{\mathbf{u}}$ denote the carrier, the partial order and the labelling of $\mathbf{u}$ respectively. The set of labelled posets is denoted $\mathsf{LP}$ and the empty labelled poset is $1$.
\end{definition}

\begin{definition}[Poset morphism]
For posets $\langle S,\leq\rangle$ and $\langle S',\leq'\rangle$ and function $\varphi:S\rightarrow S'$, $f$ is called a poset morphism if for $a_0,a_1\in S$ with $a_0\leq a_1$, then $\varphi(a_0)\leq' \varphi(a_1)$ holds.
\end{definition}

\begin{definition}[Labelled poset isomorphism]
Let $\mathbf{u}=\langle S_1,\leq_1,\lambda_1\rangle$ and $\mathbf{v}=\langle S_2,\leq_2,\lambda_2\rangle$ be labelled posets. A labelled poset morphism $\varphi$ from $\mathbf{u}=\langle S_1,\leq_1,\lambda_1\rangle$ to $\mathbf{v}=\langle S_2,\leq_2,\lambda_2\rangle$ is a poset morphism from $\langle S_1,\leq_1\rangle$ and $\langle S_2,\leq_2\rangle$ with $\lambda_2\circ \varphi=\lambda_1$. Moreover, $\varphi$ is a labelled poset isomorphism if it is a bijection with $\varphi^{-1}$ is a poset isomorphism from $\langle S_2,\leq_2,\lambda_2\rangle$ to $\langle S_1,\leq_1,\lambda_1\rangle$. We say that $\mathbf{u}=\langle S_1,\leq_1,\lambda_1\rangle$ is isomorphic to $\mathbf{v}=\langle S_2,\leq_2,\lambda_2\rangle$ denoted $\langle S_1,\leq_1,\lambda_1\rangle\sim\langle S_2,\leq_2,\lambda_2\rangle$, if there exists a poset isomorphism $\varphi$ between $\langle S_1,\leq_1,\lambda_1\rangle$ and $\langle S_2,\leq_2,\lambda_2\rangle$.
\end{definition}

\begin{definition}[Multiset]
A multiset is a kind of set of objects which may be repetitive denoted $\mset{-}$, such that $\mset{a,b,b}$ is significantly distinguishable from $\mset{a,b}$.
\end{definition}

\begin{definition}[Pomset]
A partially ordered multiset, pomset, is a $\sim$-equivalence class of posets. The set of pomsets is denoted $\mathsf{Pom}$; the empty pomset is denoted $\mathbf{I}$ and the $\sim$-equivalence class of $\mathbf{I}$ is also denoted by $\mathbf{I}$.
\end{definition}

\begin{definition}[Pomset composition in parallel]
Let $U,V\in\mathsf{Pom}$ with $U=[\mathbf{u}]$ and $V=[\mathbf{v}]$. We write $U\parallel V$ for the parallel composition of $U$ and $V$, which is the pomset represented by $\mathbf{u}\parallel\mathbf{v}$, where

\[A_{\mathbf{u}\parallel\mathbf{v}}=A_{\mathbf{u}}\cup A_{\mathbf{v}}
\quad\quad\leq_{\mathbf{u}\parallel\mathbf{v}}=\leq_{\mathbf{u}}\cup\leq_{\mathbf{v}}\]
\end{definition}

\begin{definition}[Pomset composition in sequence]
Let $U,V\in\mathsf{Pom}$ with $U=[\mathbf{u}]$ and $V=[\mathbf{v}]$. We write $U\cdot V$ for the sequential composition of $U$ and $V$, which is the pomset represented by $\mathbf{u}\cdot\mathbf{v}$, where

\[A_{\mathbf{u}\cdot\mathbf{v}}=A_{\mathbf{u}}\cup A_{\mathbf{v}}
\quad\quad\leq_{\mathbf{u}\cdot\mathbf{v}}=\leq_{\mathbf{u}}\cup\leq_{\mathbf{v}}\cup(A_{\mathbf{u}}\times A_{\mathbf{v}})\]
\end{definition}

\begin{definition}[Pomset types]
Let $U\in\mathsf{Pom}$, $U$ is sequential (resp. parallel) if there exist non-empty pomsets $U_1$ and $U_2$ such that $U=U_1\cdot U_2$ (resp. $U=U_1\parallel U_2$).
\end{definition}

\begin{definition}[Factorization]
Let $U\in\mathsf{Pom}$. (1) When $U=U_1\cdot\cdots \cdot U_i\cdot\cdots \cdot U_n$ with each $U_i$ non-sequential and non-empty, the sequence $U_1,\cdots,U_i,\cdots,U_n$ is called a sequential factorization of $U$. (2) When $U=U_1\parallel\cdots \parallel U_i\parallel\cdots \parallel U_n$ with each $U_i$ non-parallel and non-empty, the multiset $\mset{U_1,\cdots,U_i,\cdots,U_n}$ is called a parallel factorization of $U$.
\end{definition}

\begin{lemma}[Factorization]
Sequential and parallel factorizations exist uniquely.
\end{lemma}

\subsection{Equational Logic}\label{el}

\begin{definition}[Signature]
A signature $\Sigma$ consists of a finite set of function symbols (or operators) $f,g,\cdots$, where each function symbol $f$ has an arity $ar(f)$, being its number of arguments. A function symbol $\alpha,\beta,\gamma,\cdots$ of arity \emph{zero} is called a constant, a function symbol of arity one is called unary, and a function symbol of arity two is called binary.
\end{definition}

\begin{definition}[Term]
Let $\Sigma$ be a signature. The set $\mathbb{T}(\Sigma)$ of (open) terms $s,t,u,\cdots$ over $\Sigma$ is defined as the least set satisfying: (1)each variable is in $\mathbb{T}(\Sigma)$; (2) if $f\in \Sigma$ and $t_1,\cdots,t_{ar(f)}\in\mathbb{T}(\Sigma)$, then $f(t_1,\cdots,t_{ar(f)}\in\mathbb{T}(\Sigma))$. A term is closed if it does not contain free variables. The set of closed terms is denoted by $\mathcal{T}(\Sigma)$.
\end{definition}

\begin{definition}[Substitution]
Let $\Sigma$ be a signature. A substitution is a mapping $\sigma$ from variables to the set $\mathbb{T}(\Sigma)$ of open terms. A substitution extends to a mapping from open terms to open terms: the term $\sigma(t)$ is obtained by replacing occurrences of variables $x$ in t by $\sigma(x)$. A substitution $\sigma$ is closed if $\sigma(x)\in\mathcal{T}(\Sigma)$ for all variables $x$.
\end{definition}

\begin{definition}[Axiomatization]
An axiomatization over a signature $\Sigma$ is a finite set of equations, called axioms, of the form $s=t$ with $s,t\in\mathbb{T}(\Sigma)$.
\end{definition}

\begin{definition}[Equality relation]
An axiomatization over a signature $\Sigma$ induces a binary equality relation $=$ on $\mathbb{T}(\Sigma)$ as follows. (1)(Substitution) If $s=t$ is an axiom and $\sigma$ a substitution, then $\sigma(s)=\sigma(t)$. (2)(Equivalence) The relation $=$ is closed under reflexivity, symmetry, and transitivity. (3)(Context) The relation $=$ is closed under contexts: if $t=u$ and $f$ is a function symbol with $ar(f)>0$, then $f(s_1,\cdots,s_{i-1},t,s_{i+1},\cdots,s_{ar(f)})=f(s_1,\cdots,s_{i-1},u,s_{i+1},\cdots,s_{ar(f)})$.
\end{definition}

\begin{definition}[Model]
Assume an axiomatization $\mathcal{E}$ over a signature $\Sigma$, which induces an equality relation $=$. A model for $\mathcal{E}$ consists of a set $\mathcal{M}$ together with a mapping $\phi: \mathcal{T}(\Sigma)\rightarrow\mathcal{M}$. (1) $(\mathcal{M},\phi)$ is sound for $\mathcal{E}$ if $s=t$ implies $\phi(s)\equiv\phi(t)$ for $s,t\in\mathcal{T}(\Sigma)$; (2) $(\mathcal{M},\phi)$ is complete for $\mathcal{E}$ if $\phi(s)\equiv\phi(t)$ implies $s=t$ for $s,t\in\mathcal{T}(\Sigma)$.
\end{definition}

\begin{definition}[Term rewriting system]
Assume a signature $\Sigma$. A rewrite rule is an expression $s\rightarrow t$ with $s,t\in\mathbb{T}(\Sigma)$, where: (1) the left-hand side $s$ is not a single variable; (2) all variables that occur at the right-hand side $t$ also occur in the left-hand side $s$. A term rewriting system (TRS) is a finite set of rewrite rules.
\end{definition}

\begin{definition}[Rewrite relation]
A TRS over a signature $\Sigma$ induces an one-step rewrite relation $\rightarrow$ on $\mathbb{T}(\Sigma)$ as follows. (1) (Substitution) If $s\rightarrow t$ is a rewrite rule and $\sigma$ a substitution, then $\sigma(s)\rightarrow\sigma(t)$. (2) (Context) The relation $\rightarrow$ is closed under contexts: if $t\rightarrow u$ and f is a function symbol with $ar(f)>0$, then $f(s_1,\cdots,s_{i-1},t,s_{i+1},\cdots,s_{ar(f)})\rightarrow f(s_1,\cdots,s_{i-1},u,s_{i+1},\cdots,s_{ar(f)})$. The rewrite relation $\rightarrow^*$ is the reflexive transitive closure of the one-step rewrite relation $\rightarrow$: (1) if $s\rightarrow t$, then $s\rightarrow^* t$; (2) $t\rightarrow^* t$; (3) if $s\rightarrow^* t$ and $t\rightarrow^* u$, then $s\rightarrow^* u$.
\end{definition}

\begin{definition}[Normal form]
A term is called a normal form for a TRS if it cannot be reduced by any of the rewrite rules.
\end{definition}

\begin{definition}[Termination]
A TRS is terminating if it does not induce infinite reductions $t_0\rightarrow t_1\rightarrow t_2\rightarrow \cdots$.
\end{definition}

\begin{definition}[Weak confluence]
A TRS is weakly confluent if for each pair of one-step reductions $s\rightarrow t_1$ and $s\rightarrow t_2$, there is a term $u$ such that $t_1\rightarrow^* u$ and $t_2\rightarrow^* u$.
\end{definition}

\begin{theorem}[Newman's lemma]
If a TRS is terminating and weakly confluent, then it reduces each term to a unique normal form.
\end{theorem}

\begin{definition}[Convergence]
A pair of terms $s$ and $t$ is said to be convergent if there exists a term $u$ such that $s\rightarrow^* u$ and $t\rightarrow^* u$.
\end{definition}

Axiomatizations can give rise to TRSs that are not weakly confluent, which can be remedied by Knuth-Bendix completion \cite{KB}. It determines overlaps in left hand sides of rewrite rules, and introduces extra rewrite rules to join the resulting right hand sides, witch are called critical pairs.

\begin{theorem}[Weak confluence]
A TRS is weakly confluent if and only if all its critical pairs are convergent.
\end{theorem}

\begin{definition}[Elimination property]
Let a process algebra with a defined set of basic terms as a subset of the set of closed terms over the process algebra. Then the process algebra has the elimination to basic terms property if for every closed term $s$ of the algebra, there exists a basic term $t$ of the algebra such that the algebra$\vdash s=t$.
\end{definition}

\begin{definition}[Strongly normalizing]
A term $s_0$ is called strongly normalizing if does not an infinite series of reductions beginning in $s_0$.
\end{definition}

\begin{definition}
We write $s>_{lpo} t$ if $s\rightarrow^+ t$ where $\rightarrow^+$ is the transitive closure of the reduction relation defined by the transition rules of an algebra.
\end{definition}

\begin{theorem}[Strong normalization]\label{SN}
Let a term rewriting system (TRS) with finitely many rewriting rules and let $>$ be a well-founded ordering on the signature of the corresponding algebra. If $s>_{lpo} t$ for each rewriting rule $s\rightarrow t$ in the TRS, then the term rewriting system is strongly normalizing.
\end{theorem}

\subsection{Pomset Labelled Transition System}\label{plts}

\begin{definition}[Pomset labelled transition system]\label{pomlts}
A pomset labelled transition system (PLTS) is a quadruple $\langle\mathsf{Proc},\mathsf{Act},\{\xrightarrow{U}|U\in\mathsf{Pom}\},\mathsf{Pred}\rangle$, where:

\begin{enumerate}
  \item $\mathsf{Proc}$ is a set of processes, ranged over by $p,q$.
  \item $\mathsf{Act}$ is a set of actions, ranged over by $a,b$.
  \item $\mathsf{Pom}$ is the set of pomsets over $\mathsf{Act}$, ranged over by $U,V$.
  \item $\xrightarrow{U}\subseteq\mathsf{Proc}\times\mathsf{Pom}\times\mathsf{Proc}$ is called a pomset transition for every $U\in\mathsf{Pom}$. We write $p\xrightarrow{U}q$ instead of $(p,q)\in\xrightarrow{U}$, and write $p\xnrightarrow{U}$ if $p\xrightarrow{U}q$ with no state $q$. Intuitively, $p\xrightarrow{U}q$ means that state $p$ can evolve into state $q$ by the execution of pomset $U$. We see that traditional single action transition $p\xrightarrow{a}q$ with $a\in\mathsf{Act}$ is a special case of pomset transition in which the pomset is primitive.
  \item For every $P\in\mathsf{Pred}$, we write $pP$ (resp. $p\neg P$) if state $p$ satisfies (resp. does not satisfy) predicate $P$. Intuitively, $pP$ means that predicate $P$ holds in state $p$.
\end{enumerate}

The binary pomset transitions $p\xrightarrow{U}q$ and unary predicates $pP$ in a PLTS are called transitions.

Note that, by replacing $U\in\mathsf{Pom}$ by $a\in\mathsf{Act}$, we can get the definition of traditional labelled transition system (LTS). And we use $\mathsf{Act}(p)\subseteq\mathsf{Act}$ to denote the set of actions occurring in $p$.
\end{definition}

\begin{definition}[Finiteness conditions on a PLTS]
A PLTS is:

\begin{itemize}
  \item Finitely branching: if for every state $p$ there are only finitely many outgoing pomset transitions $p\xrightarrow{U}q$.
  \item Regular: if it is finitely branching and each state can reach only finitely many other states.
  \item Finite: if it is finitely branching and there is no infinite sequence of pomset transitions $p_0\xrightarrow{U_0}p_1\xrightarrow{U_1}\cdots$.
\end{itemize}
\end{definition}

A pomset transition system specification is a collection of inductive proof rules to derive the pomset transitions over states in $\mathsf{Proc}$.

\begin{definition}[Pomset transition system specification]
Let $\Sigma$ be a signature, and $t,t'\in\mathbb{T}(\Sigma)$. A pomset transition rule $\rho$ is of the form $\frac{H}{\alpha}$, where $H$ is the set of premises with positive premises $t\xrightarrow{U}t'$ and $tP$, and negative premises $t\xnrightarrow{U}$ and $t\neg P$; $\alpha$ is the conclusion with the form $t\xrightarrow{U}t'$ and $tP$, and $U\in\mathsf{Pom}$ and $P$ is a predicate. For the conclusion with the form of $t\xrightarrow{U}t'$, the left-hand side of the conclusion $t$ is called the source of $\rho$ and the right-hand side of the conclusion $t'$ is called the target of $\rho$. A transition rule is closed if it does not contain variables.

A pomset transition system specification (PTSS) is a set of pomset transition rules. A PTSS is positive if its transition rules do not contain negative premises.
\end{definition}

\begin{definition}[Pomset transitions]
Let $p,q\in\mathsf{Proc}$ and $U\in\mathsf{Pom}$, $p\xrightarrow{U} q$ is called a pomset transition from $p$ to $q$.
\end{definition}

\begin{definition}[Pomset bisimulation]
Let $\langle\mathsf{Proc},\mathsf{Act},\{\xrightarrow{U}|U\in\mathsf{Pom}\},\mathsf{Pred}\rangle$ be a PLTS, $S$ a set of states, and $p,q\in\mathsf{Proc}$ and $s,s'\in S$. A pomset bisimulation is a relation $R\subseteq\langle\mathsf{Proc},S\rangle\times\langle\mathsf{Proc},S\rangle$, such that:

\begin{enumerate}
  \item If $\langle p,s\rangle R \langle q,s\rangle$ and $\langle p,s\rangle\xrightarrow{U}\langle p',s'\rangle$, then $\langle q,s\rangle\xrightarrow{U}\langle q',s'\rangle$, with $U\in \mathsf{Pom}$, and $\langle p',s'\rangle R \langle q',s'\rangle$.
  \item If $\langle p,s\rangle R \langle q,s\rangle$ and $\langle q,s\rangle\xrightarrow{U}\langle q',s'\rangle$, then $\langle p,s\rangle\xrightarrow{U}\langle p',s'\rangle$, with $U\in \mathsf{Pom}$, and $\langle p',s'\rangle R \langle q',s'\rangle$.
  \item If $\langle p,s\rangle R \langle q,s\rangle$ and $\langle p,s\rangle P$, then $\langle q,s\rangle P$.
  \item If $\langle p,s\rangle R \langle q,s\rangle$ and $\langle q,s\rangle P$, then $\langle p,s\rangle P$.
\end{enumerate}

We say that $p$, $q$ are pomset bisimilar, written $p\pb q$, if there exists a pomset bisimulation $R$, such that $\langle p,s\rangle R \langle q,s\rangle$.
\end{definition}

\begin{definition}[Congruence]
Let $\Sigma$ be a signature. An equivalence relation $R$ on $\mathcal{T}(\Sigma)$ is a congruence if for each $f\in\Sigma$, if $s_i R t_i$ for $i\in\{1,\cdots,ar(f)\}$, then $f(s_1,\cdots,s_{ar(f)}) R f(t_1,\cdots,t_{ar(f)})$.
\end{definition}

\begin{definition}[Operational conservative extension]
$\Sigma_0$ and $\Sigma_1$ are two signatures, the union of $\Sigma_0$ and $\Sigma_1$, denoted $\Sigma_0\oplus\Sigma_1$, agrees on the arity of the function symbols in the intersection of $\Sigma_0$ and $\Sigma_1$. $T_0$ and $T_1$ are PTSSs over $\Sigma_0$ and $\Sigma_1$, the sum of $T_0$ and $T_1$, denoted $T_0\oplus T_1$, is the PTSS over signature $\Sigma_0\oplus\Sigma_1$ containing the pomset transition rules in $T_0$ and $T_1$. A PTSS $T_0\oplus T_1$ is an operational conservative extension of PTSS $T_0$, if for each closed transition rule $\frac{N}{\alpha}$ such that:

\begin{itemize}
  \item $N$ contains only negative literals.
  \item The left-hand side of $\alpha$ is in $T(\Sigma_0)$.
  \item $T_0\oplus T_1\vdash\frac{N}{\alpha}$.
\end{itemize}

we have that $T_0\vdash\frac{N}{\alpha}$.
\end{definition}

\begin{definition}[Source-dependency]
The source-dependent variables in a pomset transition rule $\rho$ are defined inductively as follows:

\begin{itemize}
  \item All variables in the source of $\rho$ are source-dependent.
  \item If $t\xrightarrow{U}t'$ is a premise of $\rho$ and all variables in $t$ are source-dependent, then all variables in $t'$ are source-dependent.
\end{itemize}

A pomset transition rule is source-dependent if all its variables are.
\end{definition}

\begin{definition}[Freshness]
We say that a term in $\mathbb{T}(\Sigma_0\oplus\Sigma_1)$ is fresh if it contains a function symbol (may be a constant) from $\Sigma_1\setminus\Sigma_0$. Similarly, a predicate symbol in $T_1$ is fresh if it does not occur in $T_0$.
\end{definition}

\begin{theorem}[Operational conservative extension]
Let $T_0$ and $T_1$ be PTSSs over signature $\Sigma_0$ and $\Sigma_1$, respectively. $T_0\oplus T_1$ is an operational conservative extension of $T_0$, if:

\begin{enumerate}
  \item Each $\rho\in T_0$ is source-dependent.
  \item For each $\rho\in T_1$,
       \begin{itemize}
         \item Either the source of $\rho$ is fresh.
         \item Or $\rho$ has a premise of the form $t\xrightarrow{U}t'$ or $tP$, where:
              \begin{itemize}
                \item $t\in\mathbb{T}(\Sigma_0)$.
                \item All variables in $t$ occur in the source of $\rho$.
                \item $t'$, $P$ and every $a\in U$ are fresh.
              \end{itemize}
       \end{itemize}
\end{enumerate}
\end{theorem} 
\newpage\section{Basic Algebra for Transformers}\label{bat}

In this section, we will discuss the algebraic laws for the set of atomic actions $\mathsf{Act}$ the set of guards $\mathsf{Grd}$, the deadlock constant $\mathbf{0}$ and the identity $\mathbf{I}$, the set of pomsets over $\mathsf{Act}\cup\{\mathbf{0}\}\cup\{\mathbf{I}\}$ denoted $\mathsf{Pom}$, and the serial composition $\circ$ and the alternative composition $+$, and let $\alpha\in\mathsf{Act}\cup\{\mathbf{0}\}\cup\{\mathbf{I}\}$, $U\in\mathsf{Pom}$. The resulted algebra is called Basic Algebra for True Concurrency, abbreviated BAT.

\subsection{Axiomatic System}

In the following, let $\alpha,\beta,\alpha',\beta'\in \mathsf{Act}$, $\phi,\psi\in\mathsf{Grd}$, and let variables $x,y,z$ range over the set of terms BAT, $p,q,s$ range over the set of closed terms. The set of axioms of BAT consists of the laws given in Table \ref{AxiomsForBAT}. The predicate $test(\phi,s)$ represents that $\phi$ holds in the state $s$, and $test(\mathbf{I},s)$ holds and $test(\mathbf{0},s)$ does not hold. $effect(U,s)\in S$ denotes $s'$ in $s\xrightarrow{U}s'$ for $U\in\mathsf{Pom}$. The predicate weakest precondition $wp(U,\phi)$ denotes that $\forall s,s'\in S, test(\phi,effect(U,s))$ holds.

We extend $\mathsf{Grd}$ to the set of basic guards $G$ with element $\phi,\psi,\cdots$, which is generated by the following formation rules:

\[\phi::=\mathbf{0}~|~\mathbf{I}~|~\neg\phi~|~\psi\in \mathsf{Grd}~|~\phi+\psi~|~\phi\circ\psi\]

The deadlock $\mathbf{0}$ is caused by incomposibility of serial composition, which may be caused by the mismatch of dimensions of vectors or matrices, we define the composibility function of two actions $a,b\in\mathsf{Act}$ as follows: \[\gamma(a,b)=\begin{cases}
                                                                                                                                                                                                                                                    a\circ b,& a,b\mbox{ are composible};\\
                                                                                                                                                                                                                                                    \mathbf{0},& \mbox{otherwise}.
                                                                                                                                                                                                                                                \end{cases}\]
Each serial composition implicitly executes this function.

\begin{center}
    \begin{longtable}{|c|c|}
      \caption{Axioms of BAT}\\
      \hline No. &Axioms\\
      \hline
      \endfirsthead
      \multicolumn{2}{c}{\bfseries continuing: Axioms of BAT} \\
      \hline No. &Axioms\\
      \hline
      \endhead
      \hline
      \multicolumn{2}{r}{to be continued\ldots} \\
      \endfoot
      \hline
      \multicolumn{2}{r}{end of Axioms of BAT} \\
      \endlastfoot
      $A1$ & $x+ y = y+ x$\\
      $A2$ & $(x+ y)+ z = x+ (y+ z)$\\
      $A3$ & $x+ x = x$\\
      $A4$ & $(x+ y)\circ z = x\circ z + y\circ z$\\
      $A5$ & $(x\circ y)\circ z = x\circ(y\circ z)$\\
      $A6$ & $x+\mathbf{0}=x$\\
      $A7$ & $\mathbf{0}\circ x=\mathbf{0}$\\
      $A8$ & $\mathbf{I}\circ x=x$\\
      $A9$ & $x\circ \mathbf{I}=x$\\
      $G1$ & $\phi\circ\neg\phi = \mathbf{0}$\\
      $G2$ & $\phi+\neg\phi = \mathbf{I}$\\
      $G3$ & $\phi\circ\mathbf{0} = \mathbf{0}$\\
      $G4$ & $\phi\circ(x+y)=\phi\circ x+\phi\circ y$\\
      $G5$ & $\phi\circ(x\circ y)= (\phi\circ x)\circ y$\\
      $G6$ & $(\phi+\psi)\circ x = \phi\circ x + \psi\circ x$\\
      $G7$ & $(\phi\circ \psi)\circ x = \phi\circ(\psi\circ x)$\\
      $G8$ & $\phi=\mathbf{I}$ if $\forall s\in S.test(\phi,s)$\\
      $G9$ & $\phi_0\circ\cdots\circ\phi_n = \mathbf{0}$ if $\forall s\in S,\exists i\leq n.test(\neg\phi_i,s)$\\
      $G10$ & $wp(U,\phi)U\phi=wp(U,\phi)U$\\
      $G11$ & $\neg wp(U,\phi)U\neg\phi=\neg wp(U,\phi)U$
      \label{AxiomsForBAT}
   \end{longtable}
\end{center}

\begin{definition}[Basic terms of BAT]\label{BTBAT}
The set of basic terms of BAT, $\mathcal{B}(BAT)$, is inductively defined as follows:

\begin{enumerate}
  \item $\mathsf{Act}\cup\{\mathbf{0}\}\cup\{\mathbf{I}\}\subset\mathcal{B}(BAT)$.
  \item $\mathbf{Grd}\subset\mathcal{B}(BAT)$.
  \item If $U\in \mathsf{Pom}, t\in\mathcal{B}(BAT)$ then $U\circ t\in\mathcal{B}(BAT)$.
  \item If $\phi\in \mathsf{Grd}, t\in\mathcal{B}(BAT)$ then $\phi\circ t\in\mathcal{B}(BAT)$.
  \item If $t,s\in\mathcal{B}(BAT)$ then $t+ s\in\mathcal{B}(BAT)$.
\end{enumerate}
\end{definition}

\begin{theorem}[Elimination theorem of BAT]\label{ETBAT}
Let $p$ be a closed BAT term. Then there is a basic BAT term $q$ such that $BAT\vdash p=q$.
\end{theorem}

\begin{proof}
(1) Firstly, suppose that the following ordering on the signature of BAT is defined: $\circ > +$ and the symbol $\circ$ is given the lexicographical status for the first argument, then for each rewrite rule $p\rightarrow q$ in Table \ref{TRSForBAT} relation $p>_{lpo} q$ can easily be proved. We obtain that the term rewrite system shown in Table \ref{TRSForBAT} is strongly normalizing, for it has finitely many rewriting rules, and $>$ is a well-founded ordering on the signature of BAT, and if $s>_{lpo} t$, for each rewriting rule $s\rightarrow t$ is in Table \ref{TRSForBAT} (see Theorem \ref{SN}).

\begin{center}
    \begin{longtable}{|c|c|}
      \caption{Rewriting rules of BAT}\\
      \hline No. &Rewriting rules\\
      \hline
      \endfirsthead
      \multicolumn{2}{c}{\bfseries continuing: Rewriting rules of BAT} \\
      \hline No. &Rewriting rules\\
      \hline
      \endhead
      \hline
      \multicolumn{2}{r}{to be continued\ldots} \\
      \endfoot
      \hline
      \multicolumn{2}{r}{end of Axioms of BAT} \\
      \endlastfoot
      $RA3$ & $x+ x \rightarrow x$\\
      $RA4$ & $(x+ y)\circ z \rightarrow x\circ z + y\circ z$\\
      $RA5$ & $(x\circ y)\circ z \rightarrow x\circ(y\circ z)$\\
      $RA6$ & $x+\mathbf{0}\rightarrow x$\\
      $RA7$ & $\mathbf{0}\circ x\rightarrow \mathbf{0}$\\
      $RA8$ & $\mathbf{I}\circ x\rightarrow x$\\
      $RA9$ & $x\circ \mathbf{I}\rightarrow x$\\
      $RG1$ & $\phi\circ\neg\phi \rightarrow \mathbf{0}$\\
      $RG2$ & $\phi+\neg\phi \rightarrow \mathbf{I}$\\
      $RG3$ & $\phi\circ\mathbf{0} \rightarrow \mathbf{0}$\\
      $RG4$ & $\phi\circ(x+y)\rightarrow\phi\circ x+\phi\circ y$\\
      $RG5$ & $\phi\circ(x\circ y)\rightarrow (\phi\circ x)\circ y$\\
      $RG6$ & $(\phi+\psi)\circ x \rightarrow \phi\circ x + \psi\circ x$\\
      $RG7$ & $(\phi\circ \psi)\circ x \rightarrow \phi\circ(\psi\circ x)$\\
      $RG8$ & $\phi\rightarrow\mathbf{I}$ if $\forall s\in S.test(\phi,s)$\\
      $RG9$ & $\phi_0\circ\cdots\circ\phi_n \rightarrow \mathbf{0}$ if $\forall s\in S,\exists i\leq n.test(\neg\phi_i,s)$\\
      $RG10$ & $wp(U,\phi)U\phi\rightarrow wp(U,\phi)U$\\
      $RG11$ & $\neg wp(U,\phi)U\neg\phi\rightarrow\neg wp(U,\phi)U$
      \label{TRSForBAT}
    \end{longtable}
\end{center}

(2) Then we prove that the normal forms of closed BAT terms are basic BAT terms.

Suppose that $p$ is a normal form of some closed $BAT$ term and suppose that $p$ is not a basic term. Let $p'$ denote the smallest sub-term of $p$ which is not a basic term. It implies that each sub-term of $p'$ is a basic term. Then we prove that $p$ is not a term in normal form. It is sufficient to induct on the structure of $p'$:

\begin{itemize}
  \item Case $p'\equiv U, U\in \mathsf{Pom}$. $p'$ is a basic term, which contradicts the assumption that $p'$ is not a basic term, so this case should not occur.
  \item Case $p'\equiv \phi, \phi\in \mathsf{Grd}$. $p'$ is a basic term, which contradicts the assumption that $p'$ is not a basic term, so this case should not occur.
  \item Case $p'\equiv p_1\circ p_2$. By induction on the structure of the basic term $p_1$:
      \begin{itemize}
        \item Subcase $p_1\in \mathsf{Pom}$. $p'$ would be a basic term, which contradicts the assumption that $p'$ is not a basic term;
        \item Subcase $p_1\in \mathsf{Grd}$. $p'$ would be a basic term, which contradicts the assumption that $p'$ is not a basic term;
        \item Subcase $p_1\equiv U\circ p_1'$. $RA5$ or $RA9$ rewriting rule can be applied. So $p$ is not a normal form;
        \item Subcase $p_1\equiv \phi\circ p_1'$. $RG1$, $RG3$, $RG4$, $RG5$, $RG7$, or $RG8-9$ rewriting rules can be applied. So $p$ is not a normal form;
        \item Subcase $p_1\equiv p_1'+ p_1''$. $RA4$, $RA6$, $RG2$, or $RG6$ rewriting rules can be applied. So $p$ is not a normal form.
      \end{itemize}
  \item Case $p'\equiv p_1+ p_2$. By induction on the structure of the basic terms both $p_1$ and $p_2$, all subcases will lead to that $p'$ would be a basic term, which contradicts the assumption that $p'$ is not a basic term.
\end{itemize}
\end{proof}

\subsection{Operational Semantics}

We give the operational transition rules for $\mathbf{I}$, atomic action $\alpha\in\mathsf{Act}$, pomset $U\in\mathsf{Pom}$, guard $\phi\in\mathsf{Grd}$ and operators $\circ$ and $+$ as Table \ref{TRForBAT} shows. And the predicate $\xrightarrow{U}\surd$ represents successful termination after execution of $U$.

\begin{center}
    \begin{table}
        \[\frac{}{\langle\mathbf{I},s\rangle\rightarrow\langle\surd,s\rangle}\]
        \[\frac{}{\langle U,s\rangle\xrightarrow{U}\langle\surd,s'\rangle}\mbox{ if }s'\in effect(U,s)\]
        \[\frac{}{\langle\phi,s\rangle\rightarrow\langle\surd,s\rangle}\mbox{ if }test(\phi,s)\]
        \[\frac{\langle x,s\rangle\xrightarrow{U}\langle\surd,s'\rangle}{\langle x+ y,s\rangle\xrightarrow{U}\langle\surd,s'\rangle} \quad\frac{\langle x,s\rangle\xrightarrow{U}\langle x',s'\rangle}{\langle x+ y,s\rangle\xrightarrow{U}\langle x',s'\rangle}\]
        \[\frac{\langle y,s\rangle\xrightarrow{U}\langle\surd,s'\rangle}{\langle x+ y,s\rangle\xrightarrow{U}\langle\surd,s'\rangle} \quad\frac{\langle y,s\rangle\xrightarrow{U}\langle y',s'\rangle}{\langle x+ y,s\rangle\xrightarrow{U}\langle y',s'\rangle}\]
        \[\frac{\langle x,s\rangle\xrightarrow{U}\langle\surd,s'\rangle}{\langle x\circ y,s\rangle\xrightarrow{U} \langle y,s'\rangle} \quad\frac{\langle x,s\rangle\xrightarrow{U}\langle x',s'\rangle}{\langle x\circ y,s\rangle\xrightarrow{U}\langle x'\circ y,s'\rangle}\]
        \caption{Transition rules of BAT}
        \label{TRForBAT}
    \end{table}
\end{center}

\begin{theorem}[Congruence of BAT with respect to pomset bisimulation equivalence]
Pomset bisimulation equivalence is a congruence with respect to BAT.
\end{theorem}

\begin{proof}
Since the PTSS of BAT in Table \ref{TRForBAT} is positive after reduction and in panth format, pomset bisimulation equivalence is a congruence with respect to BAT.
\end{proof}

\subsection{Soundness and Completeness}

\begin{theorem}[Soundness of BAT modulo pomset bisimulation equivalence]\label{SBAT}
The axiomatization of BAT is sound modulo pomset bisimulation equivalence, i.e., let $x$ and $y$ be BAT terms. If $BAT\vdash x=y$, then $x\pb y$.
\end{theorem}

\begin{proof}
Since pomset bisimulation $\pb$ is both an equivalent and a congruent relation, we only need to check if each axiom in Table \ref{AxiomsForBAT} is sound modulo pomset bisimulation equivalence $\pb$. The proof is trivial and we omit it.
\end{proof}

\begin{theorem}[Completeness of BAT modulo pomset bisimulation equivalence]\label{CBAT}
The axiomatization of BAT is complete modulo pomset bisimulation equivalence, i.e., let $p$ and $q$ be closed BAT terms, if $p\pb q$ then $p=q$.
\end{theorem}

\begin{proof}
Firstly, by the elimination theorem of BAT, we know that for each closed BAT term $p$, there exists a closed basic BAT term $p'$, such that $BAT\vdash p=p'$, so, we only need to consider closed basic BAT terms.

The basic terms (see Definition \ref{BTBAT}) modulo associativity and commutativity (AC) of conflict $+$ (defined by axioms $A1$ and $A2$ in Table \ref{AxiomsForBAT}), and this equivalence is denoted by $=_{AC}$. Then, each equivalence class $s$ modulo AC of $+$ has the following normal form

\[s_1+\cdots+ s_k\]

with each $s_i$ either an atomic event or of the form $t_1\circ t_2$, and each $s_i$ is called the summand of $s$.

Now, we prove that for normal forms $n$ and $n'$, if $n\pb n'$ then $n=_{AC}n'$. It is sufficient to induct on the sizes of $n$ and $n'$.

\begin{itemize}
  \item Consider a summand $U$ of $n$. Then $\langle n,s\rangle\xrightarrow{U}\langle \surd,s'\rangle$, so $n\sim_p n'$ implies $\langle n',s\rangle\xrightarrow{U}\langle \surd,s\rangle$, meaning that $n'$ also contains the summand $U$.
  \item Consider a summand $\phi$ of $n$. Then $\langle n,s\rangle\rightarrow\langle \surd,s\rangle$, if $test(\phi,s)$ holds, so $n\sim_p n'$ implies $\langle n',s\rangle\rightarrow\langle \surd,s\rangle$, if $test(\phi,s)$ holds, meaning that $n'$ also contains the summand $\phi$.
  \item Consider a summand $t_1\circ t_2$ of $n$. Then $\langle n,s\rangle\xrightarrow{t_1}\langle t_2,s'\rangle$, so $n\sim_p n'$ implies $\langle n',s\rangle\xrightarrow{t_1}\langle t_2',s'\rangle$ with $t_2\sim_p t_2'$, meaning that $n'$ contains a summand $t_1\circ t_2'$. Since $t_2$ and $t_2'$ are normal forms and have sizes smaller than $n$ and $n'$, by the induction hypotheses $t_2\pb t_2'$ implies $t_2=_{AC} t_2'$.
\end{itemize}

So, we get $n=_{AC} n'$.

Finally, let $s$ and $t$ be basic terms, and $s\pb t$, there are normal forms $n$ and $n'$, such that $s=n$ and $t=n'$. The soundness theorem of BAT modulo pomset bisimulation equivalence (see Theorem \ref{SBAT}) yields $s\pb n$ and $t\pb n'$, so $n\pb s\pb t\pb n'$. Since if $n\pb n'$ then $n=_{AC}n'$, $s=n=_{AC}n'=t$, as desired.
\end{proof}

\newpage\section{Algebra of Parallelism for Transformers}\label{apt}

In this section, we will discuss parallel composition. The resulted algebra is called Algebra of Parallelism for Transformers, abbreviated APT.

\subsection{Operational Semantics}

Firstly, we give the transition rules for parallelism as \cref{TRForParallel} shows.

\begin{center}
    \begin{table}
        \[\frac{x\xrightarrow{\alpha}\surd\quad y\xrightarrow{\alpha'}\surd}{x\oplus y\xrightarrow{\mset{\alpha,\alpha'}}\surd} \quad\frac{x\xrightarrow{\alpha}x'\quad y\xrightarrow{\alpha'}\surd}{x\oplus y\xrightarrow{\mset{\alpha,\alpha'}}x'}\]
        \[\frac{x\xrightarrow{\alpha}\surd\quad y\xrightarrow{\alpha'}y'}{x\oplus y\xrightarrow{\mset{\alpha,\alpha'}}y'} \quad\frac{x\xrightarrow{\alpha}x'\quad y\xrightarrow{\alpha'}y'}{x\oplus y\xrightarrow{\mset{\alpha,\alpha'}}x'\oplus y'}\]
        \caption{Transition rules of APT}
        \label{TRForParallel}
    \end{table}
\end{center}

\begin{theorem}[Congruence theorem of APT]
Pomset bisimulation equivalences $\pb$ is a congruence with respect to APT.
\end{theorem}

\begin{proof}
Since the PTSS of APT in Table \ref{TRForParallel} is positive after reduction and in panth format, pomset bisimulation equivalence $\pb$ is a congruence with respect to APT.
\end{proof}

\subsection{Axiomatic System}

We design the axioms of APT in Table \ref{AxiomsForParallelism}.

\begin{center}
    \begin{longtable}{|c|c|}
      \caption{Axioms of APT}\\
      \hline No. &Axioms\\
      \hline
      \endfirsthead
      \multicolumn{2}{c}{\bfseries continuing: Axioms of APT} \\
      \hline No. &Axioms\\
      \hline
      \endhead
      \hline
      \multicolumn{2}{r}{to be continued\ldots} \\
      \endfoot
      \hline
      \multicolumn{2}{r}{end of Axioms of APT} \\
      \endlastfoot
      $P1$ & $\alpha\oplus (\alpha'\circ y) = (\alpha\oplus \alpha')\circ y$\\
      $P2$ & $(\alpha\circ x)\oplus \alpha' = (\alpha\oplus \alpha')\circ x$\\
      $P3$ & $(\alpha\circ x)\oplus (\alpha'\circ y) = (\alpha\oplus \alpha')\circ (x\oplus y)$\\
      $P4$ & $(x+ y)\oplus z = (x\oplus z)+ (y\oplus z)$\\
      $P5$ & $x\oplus (y+z) = (x\oplus y)+ (x\oplus z)$\\
      $P6$ & $\mathbf{0}\oplus x = \mathbf{0}$\\
      $P7$ & $x\oplus \mathbf{0}= \mathbf{0}$\\
      $P8$ & $\mathbf{I}\oplus x = x$\\
      $P9$ & $x\oplus \mathbf{I} = x$\\
      $G12$ & $\phi\circ(x\oplus y) =\phi\circ x\oplus \phi\circ y$\\
      $G13$ & $\phi\oplus \mathbf{0} = \mathbf{0}$\\
      $G14$ & $\mathbf{0}\oplus \phi = \mathbf{0}$\\
      $G15$ & $\phi\oplus \mathbf{I} = \phi$\\
      $G16$ & $\mathbf{I}\oplus \phi = \phi$\\
      $G17$ & $\phi\oplus\neg\phi = \mathbf{0}$\\
      $G18$ & $\phi_0\oplus\cdots\oplus\phi_n = \mathbf{0}$ if $\forall s_0,\cdots,s_n\in S,\exists i\leq n.test(\neg\phi_i,s_0\cup\cdots\cup s_n)$
      \label{AxiomsForParallelism}
    \end{longtable}
\end{center}

\begin{definition}[Basic terms of APT]\label{BTAPT}
The set of basic terms of APT, $\mathcal{B}(APT)$, is inductively defined as follows:

\begin{enumerate}
  \item $\mathsf{Act}\cup\{\mathbf{0}\}\cup\{\mathbf{I}\}\subset\mathcal{B}(APT)$.
  \item $\mathsf{Grd}\subset\mathcal{B}(APT)$.
  \item If $U\in \mathsf{Pom}, t\in\mathcal{B}(APT)$ then $U\circ t\in\mathcal{B}(APT)$.
  \item If $\phi\in \mathsf{Grd}, t\in\mathcal{B}(APT)$ then $\phi\circ t\in\mathcal{B}(APT)$.
  \item If $t,s\in\mathcal{B}(APT)$ then $t+ s\in\mathcal{B}(APT)$.
  \item If $t,s\in\mathcal{B}(APT)$ then $t\oplus s\in\mathcal{B}(APT)$.
\end{enumerate}
\end{definition}

\begin{theorem}[Elimination theorem of parallelism]\label{ETParallelism}
Let $p$ be a closed APT term. Then there is a basic APT term $q$ such that $APT\vdash p=q$.
\end{theorem}

\begin{proof}
(1) Firstly, suppose that the following ordering on the signature of APT is defined: $\oplus > \circ > +$ and the symbol $\circ$ is given the lexicographical status for the first argument, then for each rewrite rule $p\rightarrow q$ in Table \ref{TRSForAPT} relation $p>_{lpo} q$ can easily be proved. We obtain that the term rewrite system shown in Table \ref{TRSForAPT} is strongly normalizing, for it has finitely many rewriting rules, and $>$ is a well-founded ordering on the signature of APT, and if $s>_{lpo} t$, for each rewriting rule $s\rightarrow t$ is in Table \ref{TRSForAPT} (see Theorem \ref{SN}).

\begin{center}
    \begin{longtable}{|c|c|}
      \caption{Rewriting rules of APT}\\
      \hline No. &Rewriting rules\\
      \hline
      \endfirsthead
      \multicolumn{2}{c}{\bfseries continuing: Rewriting rules of APT} \\
      \hline No. &Rewriting rules\\
      \hline
      \endhead
      \hline
      \multicolumn{2}{r}{to be continued\ldots} \\
      \endfoot
      \hline
      \multicolumn{2}{r}{end of Rewriting rules of APT} \\
      \endlastfoot
      $RP1$ & $\alpha\oplus (\alpha'\circ y) \rightarrow (\alpha\oplus \alpha')\circ y$\\
      $RP2$ & $(\alpha\circ x)\oplus \alpha' \rightarrow (\alpha\oplus \alpha')\circ x$\\
      $RP3$ & $(\alpha\circ x)\oplus (\alpha'\circ y) \rightarrow (\alpha\oplus \alpha')\circ (x\oplus y)$\\
      $RP4$ & $(x+ y)\oplus z \rightarrow (x\oplus z)+ (y\oplus z)$\\
      $RP5$ & $x\oplus (y+z) \rightarrow (x\oplus y)+ (x\oplus z)$\\
      $RP6$ & $\mathbf{0}\oplus x \rightarrow \mathbf{0}$\\
      $RP7$ & $x\oplus \mathbf{0}\rightarrow \mathbf{0}$\\
      $RP8$ & $\mathbf{I}\oplus x \rightarrow x$\\
      $RP9$ & $x\oplus \mathbf{I} \rightarrow x$\\
      $RG12$ & $\phi\circ(x\oplus y) \rightarrow\phi\circ x\oplus \phi\circ y$\\
      $RG13$ & $\phi\oplus \mathbf{0} \rightarrow \mathbf{0}$\\
      $RG14$ & $\mathbf{0}\oplus \phi \rightarrow \mathbf{0}$\\
      $RG15$ & $\phi\oplus \mathbf{I} \rightarrow \phi$\\
      $RG16$ & $\mathbf{I}\oplus \phi \rightarrow \phi$\\
      $RG17$ & $\phi\oplus\neg\phi \rightarrow \mathbf{0}$\\
      $RG18$ & $\phi_0\oplus\cdots\oplus\phi_n \rightarrow \mathbf{0}$ if $\forall s_0,\cdots,s_n\in S,\exists i\leq n.test(\neg\phi_i,s_0\cup\cdots\cup s_n)$
      \label{TRSForAPT}
    \end{longtable}
\end{center}

(2) Then we prove that the normal forms of closed APT terms are basic APT terms.

Suppose that $p$ is a normal form of some closed APT term and suppose that $p$ is not a basic APT term. Let $p'$ denote the smallest sub-term of $p$ which is not a basic APT term. It implies that each sub-term of $p'$ is a basic APT term. Then we prove that $p$ is not a term in normal form. It is sufficient to induct on the structure of $p'$:

\begin{itemize}
  \item Case $p'\equiv U, U\in \mathsf{Pom}$. $p'$ is a basic APT term, which contradicts the assumption that $p'$ is not a basic APT term, so this case should not occur.
  \item Case $p'\equiv \phi, \phi\in \mathsf{Grd}$. $p'$ is a basic term, which contradicts the assumption that $p'$ is not a basic term, so this case should not occur.
  \item Case $p'\equiv p_1\circ p_2$. By induction on the structure of the basic APT term $p_1$:
      \begin{itemize}
        \item Subcase $p_1\in \mathsf{Pom}$. $p'$ would be a basic APT term, which contradicts the assumption that $p'$ is not a basic APT term;
        \item Subcase $p_1\in \mathsf{Grd}$. $p'$ would be a basic term, which contradicts the assumption that $p'$ is not a basic term;
        \item Subcase $p_1\equiv U\circ p_1'$. $RA5$ or $RA9$ rewriting rules in Table \ref{TRSForBAT} can be applied. So $p$ is not a normal form;
        \item Subcase $p_1\equiv \phi\cdot p_1'$. $RG1$, $RG3$, $RG4$, $RG5$, $RG7$, or $RG8-9$ rewriting rules can be applied. So $p$ is not a normal form;
        \item Subcase $p_1\equiv p_1'+ p_1''$. $RA4$, $RA6$, $RG2$, or $RG6$ rewriting rules in Table \ref{TRSForBAT} can be applied. So $p$ is not a normal form;
        \item Subcase $p_1\equiv p_1'\oplus p_1''$. $RP1$-$RP9$ rewrite rules in Table \ref{TRSForAPT} can be applied. So $p$ is not a normal form.
      \end{itemize}
  \item Case $p'\equiv p_1+ p_2$. By induction on the structure of the basic $APT$ terms both $p_1$ and $p_2$, all subcases will lead to that $p'$ would be a basic $APT$ term, which contradicts the assumption that $p'$ is not a basic $APT$ term.
  \item Case $p'\equiv p_1\oplus p_2$. By induction on the structure of the basic $APT$ terms both $p_1$ and $p_2$, all subcases will lead to that $p'$ would be a basic $APT$ term, which contradicts the assumption that $p'$ is not a basic $APT$ term.
\end{itemize}
\end{proof}

\subsection{Soundness and Completeness}

\begin{theorem}[Generalization of APT with respect to BAT]
APT is a generalization of BAT.
\end{theorem}

\begin{proof}
It follows from the following three facts.

\begin{enumerate}
  \item The transition rules of BAT Table \ref{TRForBAT} are all source-dependent;
  \item The sources of the transition rules for APT in Table \ref{TRForParallel} contain an occurrence of $\oplus$;
  \item The transition rules of APT are all source-dependent.
\end{enumerate}

So, APT is a generalization of BAT, i.e., BAT is an embedding of APT, as desired.
\end{proof}

\begin{theorem}[Soundness of APT modulo pomset bisimulation equivalence]\label{SAPT}
The axiomatization of APT is sound modulo pomset bisimulation equivalence, i.e., let $x$ and $y$ be APT terms. If $APT\vdash x=y$, then $x\pb y$.
\end{theorem}

\begin{proof}
Since pomset bisimulations $\pb$ is both an equivalent and a congruent relation, we only need to check if each axiom in Table \ref{AxiomsForParallelism} is sound modulo pomset bisimulation equivalence $\pb$. The proof is trivial and we omit it.
\end{proof}

\begin{theorem}[Completeness of APT modulo pomset bisimulation equivalence]\label{CAPT}
The axiomatization of APT is complete modulo pomset bisimulation equivalence, i.e., let $p$ and $q$ be closed APT terms, if $p\pb q$ then $p=q$.
\end{theorem}

\begin{proof}
Firstly, by the elimination theorem of APT (see Theorem \ref{ETParallelism}), we know that for each closed APT term $p$, there exists a closed basic APT term $p'$, such that $APT\vdash p=p'$, so, we only need to consider closed basic APT terms.

The basic terms (see Definition \ref{BTAPT}) modulo associativity and commutativity (AC) of conflict $+$ (defined by axioms $A1$ and $A2$ in Table \ref{AxiomsForBAT}), and these equivalences is denoted by $=_{AC}$. Then, each equivalence class $s$ modulo AC of $+$ has the following normal form

\[s_1+\cdots+ s_k\]

with each $s_i$ either an atomic event or of the form

\[t_1\circ\cdots\circ t_m\]

with each $t_j$ either an atomic event or of the form

\[u_1\oplus\cdots\oplus u_{\ell}\]

with each $u_{\ell}$ an atomic event, and each $s_i$ is called the summand of $s$.

Now, we prove that for normal forms $n$ and $n'$, if $n\pb n'$ then $n=_{AC}n'$. It is sufficient to induct on the sizes of $n$ and $n'$.

\begin{itemize}
  \item Consider a summand $U$ of $n$. Then $\langle n,s\rangle\xrightarrow{U}\langle \surd,s'\rangle$, so $n\pb n'$ implies $\langle n',s\rangle\xrightarrow{U}\langle \surd,s\rangle$, meaning that $n'$ also contains the summand $U$.
  \item Consider a summand $\phi$ of $n$. Then $\langle n,s\rangle\rightarrow\langle \surd,s\rangle$, if $test(\phi,s)$ holds, so $n\pb n'$ implies $\langle n',s\rangle\rightarrow\langle \surd,s\rangle$, if $test(\phi,s)$ holds, meaning that $n'$ also contains the summand $\phi$.
  \item Consider a summand $t_1\cdot t_2$ of $n$,
  \begin{itemize}
    \item if $t_1\equiv \alpha'$, then $\langle n,s\rangle\xrightarrow{\alpha'}\langle t_2,s'\rangle$, so $n\pb n'$ implies $\langle n',s\rangle\xrightarrow{\alpha'}\langle t_2',s'\rangle$ with $t_2\pb t_2'$, meaning that $n'$ contains a summand $\alpha'\cdot t_2'$. Since $t_2$ and $t_2'$ are normal forms and have sizes smaller than $n$ and $n'$, by the induction hypotheses if $t_2\pb t_2'$ then $t_2=_{AC} t_2'$;
    \item if $t_1\equiv \phi'$, then $\langle n,s\rangle\rightarrow\langle t_2,s\rangle$, if $test(\phi',s)$ holds, so $n\pb n'$ implies $\langle n',s\rangle\rightarrow\langle t_2',s\rangle$ with $t_2\pb t_2'$, if $test(\phi',s)$ holds, meaning that $n'$ contains a summand $\phi'\cdot t_2'$. Since $t_2$ and $t_2'$ are normal forms and have sizes smaller than $n$ and $n'$, by the induction hypotheses if $t_2\pb t_2'$ then $t_2=_{AC} t_2'$;
    \item if $t_1\equiv \alpha_1\oplus\cdots\oplus \alpha_{ell}$, then $\langle n,s\rangle\xrightarrow{\{\alpha_1,\cdots,\alpha_{ell}\}}\langle t_2,s'\rangle$, so $n\pb n'$ implies $\langle n',s\rangle\xrightarrow{\{\alpha_1,\cdots,\alpha_{ell}\}}\langle t_2',s'\rangle$ with $t_2\pb t_2'$, meaning that $n'$ contains a summand $(\alpha_1\oplus\cdots\oplus \alpha_{ell})\cdot t_2'$. Since $t_2$ and $t_2'$ are normal forms and have sizes smaller than $n$ and $n'$, by the induction hypotheses if $t_2\pb t_2'$ then $t_2=_{AC} t_2'$;
    \item if $t_1\equiv \phi_1\oplus\cdots\oplus \phi_{ell}$, then $\langle n,s\rangle\rightarrow\langle t_2,s\rangle$, if $test(\phi_1,s),\cdots,test(\phi_{ell},s)$ hold, so $n\pb n'$ implies $\langle n',s\rangle\rightarrow\langle t_2',s\rangle$ with $t_2\pb t_2'$, if $test(\phi_1,s),\cdots,test(\phi_{ell},s)$ hold, meaning that $n'$ contains a summand $(\phi_1\oplus\cdots\oplus \phi_{ell})\cdot t_2'$. Since $t_2$ and $t_2'$ are normal forms and have sizes smaller than $n$ and $n'$, by the induction hypotheses if $t_2\pb t_2'$ then $t_2=_{AC} t_2'$.
  \end{itemize}
\end{itemize}

So, we get $n=_{AC} n'$.

Finally, let $s$ and $t$ be basic APT terms, and $s\pb t$, there are normal forms $n$ and $n'$, such that $s=n$ and $t=n'$. The soundness theorem of parallelism modulo pomset bisimulation equivalence (see Theorem \ref{SAPT}) yields $s\pb n$ and $t\pb n'$, so $n\pb s\pb t\pb n'$. Since if $n\pb n'$ then $n=_{AC}n'$, $s=n=_{AC}n'=t$, as desired.
\end{proof}

\newpage\section{Recursion}\label{rec}

In this section, we introduce recursion to capture infinite computations based on APT.

In the following, $E,F,G$ are recursion specifications, $X,Y,Z$ are recursive variables.

\begin{definition}[Recursive specification]
A recursive specification is a finite set of recursive equations

\[X_1=t_1(X_1,\cdots,X_n)\]
\[\cdots\]
\[X_n=t_n(X_1,\cdots,X_n)\]

where the left-hand sides of $X_i$ are called recursion variables, and the right-hand sides $t_i(X_1,\cdots,X_n)$ are process terms in APT with possible occurrences of the recursion variables $X_1,\cdots,X_n$.
\end{definition}

\begin{definition}[Solution]
Processes $p_1,\cdots,p_n$ are solutions for a recursive specification $\{X_i=t_i(X_1,\cdots,X_n)|i\in\{1,\cdots,n\}\}$ (with respect to pomset bisimulation equivalence $\pb$) if $p_i\pb t_i(p_1,\cdots,p_n)$ for $i\in\{1,\cdots,n\}$.
\end{definition}

\begin{definition}[Guarded recursive specification]
A recursive specification

\[X_1=t_1(X_1,\cdots,X_n)\]
\[...\]
\[X_n=t_n(X_1,\cdots,X_n)\]

is guarded if the right-hand sides of its recursive equations can be adapted to the form by applications of the axioms in APT and replacing recursion variables by the right-hand sides of their recursive equations, and there does not exist an infinite sequence of $\mathbf{I}$-transitions $\langle X|E\rangle\rightarrow\langle X'|E\rangle\rightarrow\langle X''|E\rangle\rightarrow\cdots$,

\[U_1\circ s_1(X_1,\cdots,X_n)+\cdots+U_k\circ s_k(X_1,\cdots,X_n)+V_1+\cdots+V_{\ell}\]

where $U_1,\cdots,U_k,V_1,\cdots,V_{\ell}\in \mathsf{Pom}$, and the sum above is allowed to be empty, in which case it represents the deadlock $\mathbf{0}$.
\end{definition}

\begin{definition}[Linear recursive specification]\label{LRS}
A recursive specification is linear if its recursive equations are of the form

\[U_1\circ X_1+\cdots+U_k\circ X_k+V_1+\cdots+V_{\ell}\]

where $U_1,\cdots,U_k,V_1,\cdots,V_{\ell}\in \mathsf{Pom}$, and the sum above is allowed to be empty, in which case it represents the deadlock $\mathbf{0}$.
\end{definition}

For a guarded recursive specifications $E$ with the form

\[X_1=t_1(X_1,\cdots,X_n)\]
\[\cdots\]
\[X_n=t_n(X_1,\cdots,X_n)\]

the behavior of the solution $\langle X_i|E\rangle$ for the recursion variable $X_i$ in $E$, where $i\in\{1,\cdots,n\}$, is exactly the behavior of their right-hand sides $t_i(X_1,\cdots,X_n)$, which is captured by the two transition rules in Table \ref{TRForGR}.

\subsection{Operational Semantics}

\begin{center}
    \begin{table}
        \[\frac{t_i(\langle X_1|E\rangle,\cdots,\langle X_n|E\rangle)\xrightarrow{U}\surd}{\langle X_i|E\rangle\xrightarrow{U}\surd}\]
        \[\frac{t_i(\langle X_1|E\rangle,\cdots,\langle X_n|E\rangle)\xrightarrow{U} y}{\langle X_i|E\rangle\xrightarrow{U} y}\]
        \caption{Transition rules of guarded recursion}
        \label{TRForGR}
    \end{table}
\end{center}

\begin{theorem}[Conservativity of APT with guarded recursion]
APT with guarded recursion is a conservative extension of APT.
\end{theorem}

\begin{proof}
Since the transition rules of APT are source-dependent, and the transition rules for guarded recursion in Table \ref{TRForGR} contain only a fresh constant in their source, so the transition rules of APT with guarded recursion are conservative extensions of those of APT.
\end{proof}

\begin{theorem}[Congruence theorem of APT with guarded recursion]
Pomset bisimulation equivalence $\pb$ is a congruences with respect to APT with guarded recursion.
\end{theorem}

\begin{proof}
Since the PTSS of APT with guarded recursion in Table \ref{TRForGR} is positive after reduction and in panth format, pomset bisimulation equivalence is a congruence with respect to APT with guarded recursion.
\end{proof}

\subsection{Axiomatic System}

The $RDP$ (Recursive Definition Principle) and the $RSP$ (Recursive Specification Principle) are shown in Table \ref{RDPRSP}.

\begin{center}
    \begin{longtable}{|c|c|}
      \caption{Recursive definition and specification principle}\\
      \hline No. &Axioms\\
      \hline
      \endfirsthead
      \multicolumn{2}{c}{\bfseries continuing: Recursive definition and specification principle} \\
      \hline No. &Axioms\\
      \hline
      \endhead
      \hline
      \multicolumn{2}{r}{to be continued\ldots} \\
      \endfoot
      \hline
      \multicolumn{2}{r}{end of Recursive definition and specification principle} \\
      \endlastfoot
      $RDP$ & $\langle X_i|E\rangle = t_i(\langle X_1|E\rangle,\cdots,\langle X_n|E\rangle)\quad (i\in\{1,\cdots,n\})$\\
      $RSP$ & if $y_i=t_i(y_1,\cdots,y_n)$ for $i\in\{1,\cdots,n\}$, then $y_i=\langle X_i|E\rangle \quad(i\in\{1,\cdots,n\})$
      \label{RDPRSP}
   \end{longtable}
\end{center}

$RDP$ follows immediately from the two transition rules for guarded recursion, which express that $\langle X_i|E\rangle$ and $t_i(\langle X_1|E\rangle,\cdots,\langle X_n|E\rangle)$ have the same initial transitions for $i\in\{1,\cdots,n\}$. $RSP$ follows from the fact that guarded recursive specifications have only one solution.

\subsection{Soundness and Completeness}

\begin{theorem}[Soundness of APT with guarded recursion]\label{SAPTR}
Let $x$ and $y$ be APT with guarded recursion terms. If $APT\mbox{ with guarded recursion}\vdash x=y$, then $x\pb y$.
\end{theorem}

\begin{proof}
Since pomset bisimulation equivalence is both an equivalent and a congruent relation with respect to APT with guarded recursion, we only need to check if each axiom in Table \ref{RDPRSP} is sound modulo pomset bisimulation equivalence. The proof is trivial and we omit it.
\end{proof}

\begin{theorem}[Completeness of APT with linear recursion]\label{CAPTR}
Let $p$ and $q$ be closed APT with linear recursion terms, then, if $p\pb q$ then $p=q$.
\end{theorem}

\begin{proof}
Firstly, we know that each process term in APT with linear recursion is equal to a process term $\langle X_1|E\rangle$ with $E$ a linear recursive specification.

Let $E_1$ consist of recursive equations $X=t_X$ for $X\in \mathcal{X}$ and $E_2$
consists of recursion equations $Y=t_Y$ for $Y\in\mathcal{Y}$. Let the linear recursive specification $E$ consist of recursion equations $Z_{XY}=t_{XY}$, and $\langle X|E_1\rangle\sim_s\langle Y|E_2\rangle$, and $t_{XY}$ consists of the following summands:

\begin{enumerate}
  \item $t_{XY}$ contains a summand $(\alpha_1\oplus\cdots\oplus \alpha_m)Z_{X'Y'}$ iff $t_X$ contains the summand $(\alpha_1\oplus\cdots\oplus \alpha_m)X'$ and $t_Y$ contains the summand $(\alpha_1\oplus\cdots\oplus \alpha_m)Y'$ such that $\langle X'|E_1\rangle\sim_s\langle Y'|E_2\rangle$;
  \item $t_{XY}$ contains a summand $\beta_1\oplus\cdots\oplus \beta_n$ iff $t_X$ contains the summand $\beta_1\oplus\cdots\oplus \beta_n$ and $t_Y$ contains the summand $\beta_1\oplus\cdots\oplus \beta_n$.
\end{enumerate}

Let $\sigma$ map recursion variable $X$ in $E_1$ to $\langle X|E_1\rangle$, and let $\psi$ map recursion variable $Z_{XY}$ in $E$ to $\langle X|E_1\rangle$. So, $\sigma((\alpha_1\oplus\cdots\oplus \alpha_m)X')\equiv(\alpha_1\oplus\cdots\oplus \alpha_m)\langle X'|E_1\rangle\equiv\psi((\alpha_1\oplus\cdots\oplus \alpha_m)Z_{X'Y'})$, so by $RDP$, we get $\langle X|E_1\rangle=\sigma(t_X)=\psi(t_{XY})$. Then by $RSP$, $\langle X|E_1\rangle=\langle Z_{XY}|E\rangle$, particularly, $\langle X_1|E_1\rangle=\langle Z_{X_1Y_1}|E\rangle$. Similarly, we can obtain $\langle Y_1|E_2\rangle=\langle Z_{X_1Y_1}|E\rangle$. Finally, $\langle X_1|E_1\rangle=\langle Z_{X_1Y_1}|E\rangle=\langle Y_1|E_2\rangle$, as desired.
\end{proof} 
\newpage\section{Abstraction}\label{abs}

To abstract away from the internal implementations of a transformer, and verify that the transformer exhibits the desired external behaviors, the silent step $\tau$ and abstraction operator $\tau_I$ are introduced, where $I\subseteq \mathsf{Act}\cup \mathsf{Grd}$ denotes the internal events or guards.

\subsection{Silent Step}

The silent step $\tau$ represents the internal actions, and $\tau_{\phi}$ for internal guards, when we consider the external behaviors of a process, $\tau$ steps can be removed, i.e., $\tau$ steps must keep silent. The transition rule of $\tau$ is shown in Table \ref{TRForTauG}. In the following, let the atomic action $a$ range over $\mathsf{Act}\cup\{\mathbf{0}\}\cup\{\mathbf{I}\}\cup\{\tau\}$, and $\phi$ range over $\mathsf{Grd}\cup \{\tau\}$. We use $\tau(s)$ to denote $effect(\tau,s)$, for the fact that $\tau$ only change the state of internal data environment, that is, for the external data environments, $s=\tau(s)$.

\subsubsection{Operational Semantics}

\begin{center}
    \begin{table}
        \[\frac{}{\langle\tau_{\phi},s\rangle\rightarrow\langle\surd,s\rangle}\mbox{ if }test(\tau_{\phi},s)\]
        \[\frac{}{\langle\tau,s\rangle\xrightarrow{\tau}\langle\surd,\tau(s)\rangle}\]
        \caption{Transition rule of the silent step}
        \label{TRForTauG}
    \end{table}
\end{center}

\begin{definition}[Weak pomset transitions]
Let $p,q\in\mathsf{Proc}$ and $U\in\mathsf{Pom}$, $p\xRightarrow{U} q$ is called a weak pomset transition from $p$ to $q$, where we define $\xRightarrow{a}\triangleq\xrightarrow{\tau^*}\xrightarrow{a}\xrightarrow{\tau^*}$. And $\xRightarrow{U}\triangleq\xrightarrow{\tau^*}\xrightarrow{U}\xrightarrow{\tau^*}$, for every $a\in U$. 
\end{definition}

\begin{definition}[Branching pomset bisimulation]
Assume a special termination predicate $\downarrow$, and let $\surd$ represent a state with $\surd\downarrow$. Let $S$ be a set of states, $p,q\in\mathsf{Proc}$, $U\in\mathsf{Pom}$ and $s,s'\in S$. A branching pomset bisimulation is a relation $R\subseteq\langle\mathsf{Proc},S\rangle\times\langle\mathsf{Proc},S\rangle$, such that:

 \begin{enumerate}
   \item if $(\langle p,s\rangle,\langle q,s\rangle)\in R$, and $\langle p,s\rangle\xrightarrow{U}\langle p',s'\rangle$ then
   \begin{itemize}
     \item either $U\equiv \tau^*$, and $(\langle p',s'\rangle,\langle q,s\rangle)\in R$ with $s'\in \tau(s)$;
     \item or there is a sequence of (zero or more) $\tau$-transitions $\langle q,s\rangle\xrightarrow{\tau^*} \langle q^0,s^0\rangle$, such that $(\langle p,s\rangle,\langle q^0,s^0\rangle)\in R$ and $\langle q^0,s^0\rangle\xRightarrow{U}\langle q',s'\rangle$ with $(\langle p',s'\rangle,\langle q',s'\rangle)\in R$;
   \end{itemize}
   \item if $(\langle p,s\rangle,\langle q,s\rangle)\in R$, and $\langle q,s\rangle\xrightarrow{U}\langle q',s'\rangle$ then
   \begin{itemize}
     \item either $U\equiv \tau^*$, and $(\langle p,s\rangle,\langle q',s'\rangle)\in R$;
     \item or there is a sequence of (zero or more) $\tau$-transitions $\langle p,s\rangle\xrightarrow{\tau^*} \langle p^0,s^0\rangle$, such that $(\langle p^0,s^0\rangle,\langle q,s\rangle)\in R$ and $\langle p^0,s^0\rangle\xRightarrow{U}\langle p',s'\rangle$ with $(\langle p',s'\rangle,\langle q',s'\rangle)\in R$;
   \end{itemize}
   \item if $(\langle p,s\rangle,\langle q,s\rangle)\in R$ and $\langle p,s\rangle\downarrow$, then there is a sequence of (zero or more) $\tau$-transitions $\langle q,s\rangle\xrightarrow{\tau^*}\langle q^0,s^0\rangle$ such that $(\langle p,s\rangle,\langle q^0,s^0\rangle)\in R$ and $\langle q^0,s^0\rangle\downarrow$;
   \item if $(\langle p,s\rangle,\langle q,s\rangle)\in R$ and $\langle q,s\rangle\downarrow$, then there is a sequence of (zero or more) $\tau$-transitions $\langle p,s\rangle\xrightarrow{\tau^*}\langle p^0,s^0\rangle$ such that $(\langle p^0,s^0\rangle,\langle q,s\rangle)\in R$ and $\langle p^0,s^0\rangle\downarrow$.
 \end{enumerate}

We say that $p,q$ are branching pomset bisimilar, written $p\pbb q$, if there exists a branching pomset bisimulation $R$, such that $(\langle p,s\rangle,\langle q,s\rangle)\in R$.
\end{definition}

\begin{definition}[Rooted branching pomset, step bisimulation]\label{RBPSBG}
Assume a special termination predicate $\downarrow$, and let $\surd$ represent a state with $\surd\downarrow$. Let $S$ be a set of states, $p,q\in\mathsf{Proc}$, $U\in\mathsf{Pom}$ and $s,s'\in S$. A rooted branching pomset bisimulation is a relation $R\subseteq\langle\mathsf{Proc},S\rangle\times\langle\mathsf{Proc},S\rangle$, such that:

 \begin{enumerate}
   \item if $(\langle p,s\rangle,\langle q,s\rangle)\in R$, and $\langle p,s\rangle\xrightarrow{U}\langle p',s'\rangle$ then $\langle q,s\rangle\xrightarrow{U}\langle q',s'\rangle$ with $\langle p',s'\rangle\pbb\langle q',s'\rangle$;
   \item if $(\langle p,s\rangle,\langle q,s\rangle)\in R$, and $\langle q,s\rangle\xrightarrow{U}\langle q',s'\rangle$ then $\langle p,s\rangle\xrightarrow{U}\langle p',s'\rangle$ with $\langle p',s'\rangle\pbb\langle q',s'\rangle$;
   \item if $(\langle p,s\rangle,\langle q,s\rangle)\in R$ and $\langle p,s\rangle\downarrow$, then $\langle q,s\rangle\downarrow$;
   \item if $(\langle p,s\rangle,\langle q,s\rangle)\in R$ and $\langle q,s\rangle\downarrow$, then $\langle p,s\rangle\downarrow$.
 \end{enumerate}

We say that $p,q$ are rooted branching pomset bisimilar, written $p\pbrb q$, if there exists a rooted branching pomset bisimulation $R$, such that $(\langle p,s\rangle,\langle q,s\rangle)\in R$.
\end{definition}

The silent step $\tau$ as an atomic event, is introduced into $E$. Considering the recursive specification $X=\tau X$, $\tau s$, $\tau\tau s$, and $\tau\cdots s$ are all its solutions, that is, the solutions make the existence of $\tau$-loops which cause unfairness. To prevent $\tau$-loops, we extend the definition of linear recursive specification (Definition \ref{LRS}) to the guarded one.

\begin{definition}[Guarded linear recursive specification]\label{GLRS}
A recursive specification is linear if its recursive equations are of the form

\[U_1\circ X_1+\cdots+U_k\circ X_k+V_1+\cdots+V_{\ell}\]

where $U_1,\cdots,U_k,V_1,\cdots,V_{\ell}\in \mathsf{Pom}$, and the sum above is allowed to be empty, in which case it represents the deadlock $\mathbf{0}$.

A linear recursive specification $E$ is guarded if there does not exist an infinite sequence of $\mathbf{I}$-transitions $\langle X|E\rangle\rightarrow\langle X'|E\rangle\rightarrow\langle X''|E\rangle\rightarrow\cdots$ and $\tau$-transitions $\langle X|E\rangle\xrightarrow{\tau}\langle X'|E\rangle\xrightarrow{\tau}\langle X''|E\rangle\xrightarrow{\tau}\cdots$.
\end{definition}

\begin{theorem}[Conservativity of APT with silent step and guarded linear recursion]
APT with silent step and guarded linear recursion is a conservative extension of APT with linear recursion.
\end{theorem}

\begin{proof}
Since the transition rules of APT with linear recursion are source-dependent, and the transition rules for silent step in Table \ref{TRForTauG} contain only a fresh constant $\tau$ in their source, so the transition rules of APT with silent step and guarded linear recursion is a conservative extension of those of APT with linear recursion.
\end{proof}

\begin{theorem}[Congruence theorem of APT with silent step and guarded linear recursion]
Rooted branching pomset bisimulation equivalence $\pbrb$ is a congruence with respect to APT with silent step and guarded linear recursion.
\end{theorem}

\begin{proof}
Since the PTSS of APT with silent step and guarded linear recursion in Table \ref{TRForTauG} is positive after reduction and in RBB cool format, rooted branching pomset bisimulation equivalence $\pbrb$ is a congruence with respect to APT with with silent step and guarded linear recursion.
\end{proof}

\subsubsection{Axiomatic System}

We design the axioms for the silent step $\tau$ in Table \ref{AxiomsForTau}.

\begin{center}
    \begin{longtable}{|c|c|}
      \caption{Axioms of silent step}\\
      \hline No. &Axioms\\
      \hline
      \endfirsthead
      \multicolumn{2}{c}{\bfseries continuing: Axioms of silent step} \\
      \hline No. &Axioms\\
      \hline
      \endhead
      \hline
      \multicolumn{2}{r}{to be continued\ldots} \\
      \endfoot
      \hline
      \multicolumn{2}{r}{end of Axioms of silent step} \\
      \endlastfoot
      $B1$ & $U\circ\tau=U$\\
      $B2$ & $U\circ(\tau\circ(x+y)+x)=U\circ(x+y)$\\
      $B3$ & $x\oplus\tau=x$\\
      $B4$ & $\tau\oplus x=x$\\
      $G26$ & $\tau_{\phi}\cdot x=x$\\
      $G19$ & $x\cdot \tau_{\phi} = x$\\
      $G20$ & $x\parallel\tau_{\phi} = x$\\
      $G21$ & $\tau_{\phi}\parallel x = x$
      \label{AxiomsForTau}
   \end{longtable}
\end{center}

\subsubsection{Soundness and Completeness}

\begin{theorem}[Soundness of APT with silent step and guarded linear recursion]\label{SAPTTAU}
Let $x$ and $y$ be APT with silent step and guarded linear recursion terms. If APT with silent step and guarded linear recursion $\vdash x=y$, then $x\pbrb y$.
\end{theorem}

\begin{proof}
Since rooted branching pomset bisimulation $\pbrb$ is both an equivalent and a congruent relation with respect to APT with silent step and guarded linear recursion, we only need to check if each axiom in Table \ref{AxiomsForTau} is sound modulo rooted branching pomset bisimulation equivalence. The proof is trivial and we omit it.
\end{proof}

\begin{theorem}[Completeness of APT with silent step and guarded linear recursion]\label{CAPTTAU}
Let $p$ and $q$ be closed APT with silent step and guarded linear recursion terms, then, if $p\pbrb q$ then $p=q$.
\end{theorem}

\begin{proof}
Firstly, we know that each process term in APT with silent step and guarded linear recursion is equal to a process term $\langle X_1|E\rangle$ with $E$ a guarded linear recursive specification.

The recursive equation $W=\tau+\cdots+\tau$ with $W\nequiv X_1$ in $E_1$ and $E_2$, can be removed, and the corresponding summands $aW$ are replaced by $\alpha$, to get $E_1'$ and $E_2'$, by use of the axioms $RDP$, $A3$ and $B1$, and $\langle X|E_1\rangle = \langle X|E_1'\rangle$, $\langle Y|E_2\rangle = \langle Y|E_2'\rangle$.

Let $E_1$ consists of recursive equations $X=t_X$ for $X\in \mathcal{X}$ and $E_2$
consists of recursion equations $Y=t_Y$ for $Y\in\mathcal{Y}$, and are not the form $\tau+\cdots+\tau$. Let the guarded linear recursive specification $E$ consists of recursion equations $Z_{XY}=t_{XY}$, and $\langle X|E_1\rangle\approx_{rbs}\langle Y|E_2\rangle$, and $t_{XY}$ consists of the following summands:

\begin{enumerate}
  \item $t_{XY}$ contains a summand $(\alpha_1\oplus\cdots\oplus \alpha_m)Z_{X'Y'}$ iff $t_X$ contains the summand $(\alpha_1\oplus\cdots\oplus \alpha_m)X'$ and $t_Y$ contains the summand $(\alpha_1\oplus\cdots\oplus \alpha_m)Y'$ such that $\langle X'|E_1\rangle\approx_{rbs}\langle Y'|E_2\rangle$;
  \item $t_{XY}$ contains a summand $\beta_1\oplus\cdots\oplus \beta_n$ iff $t_X$ contains the summand $\beta_1\oplus\cdots\oplus \beta_n$ and $t_Y$ contains the summand $\beta_1\oplus\cdots\oplus \beta_n$;
  \item $t_{XY}$ contains a summand $\tau Z_{X'Y}$ iff $XY\nequiv X_1Y_1$, $t_X$ contains the summand $\tau X'$, and $\langle X'|E_1\rangle\approx_{rbs}\langle Y|E_2\rangle$;
  \item $t_{XY}$ contains a summand $\tau Z_{XY'}$ iff $XY\nequiv X_1Y_1$, $t_Y$ contains the summand $\tau Y'$, and $\langle X|E_1\rangle\approx_{rbs}\langle Y'|E_2\rangle$.
\end{enumerate}

Since $E_1$ and $E_2$ are guarded, $E$ is guarded. Constructing the process term $u_{XY}$ consist of the following summands:

\begin{enumerate}
  \item $u_{XY}$ contains a summand $(\alpha_1\oplus\cdots\oplus \alpha_m)\langle X'|E_1\rangle$ iff $t_X$ contains the summand $(\alpha_1\oplus\cdots\oplus \alpha_m)X'$ and $t_Y$ contains the summand $(\alpha_1\oplus\cdots\oplus \alpha_m)Y'$ such that $\langle X'|E_1\rangle\approx_{rbs}\langle Y'|E_2\rangle$;
  \item $u_{XY}$ contains a summand $\beta_1\oplus\cdots\oplus \beta_n$ iff $t_X$ contains the summand $\beta_1\oplus\cdots\oplus \beta_n$ and $t_Y$ contains the summand $\beta_1\oplus\cdots\oplus \beta_n$;
  \item $u_{XY}$ contains a summand $\tau \langle X'|E_1\rangle$ iff $XY\nequiv X_1Y_1$, $t_X$ contains the summand $\tau X'$, and $\langle X'|E_1\rangle\approx_{rbs}\langle Y|E_2\rangle$.
\end{enumerate}

Let the process term $s_{XY}$ be defined as follows:

\begin{enumerate}
  \item $s_{XY}\triangleq\tau\langle X|E_1\rangle + u_{XY}$ iff $XY\nequiv X_1Y_1$, $t_Y$ contains the summand $\tau Y'$, and $\langle X|E_1\rangle\approx_{rbs}\langle Y'|E_2\rangle$;
  \item $s_{XY}\triangleq\langle X|E_1\rangle$, otherwise.
\end{enumerate}

So, $\langle X|E_1\rangle=\langle X|E_1\rangle+u_{XY}$, and $(\alpha_1\oplus\cdots\oplus \alpha_m)(\tau\langle X|E_1\rangle+u_{XY})=(\alpha_1\oplus\cdots\oplus \alpha_m)((\tau\langle X|E_1\rangle+u_{XY})+u_{XY})=(\alpha_1\oplus\cdots\oplus \alpha_m)(\langle X|E_1\rangle+u_{XY})=(\alpha_1\oplus\cdots\oplus \alpha_m)\langle X|E_1\rangle$, hence, $(\alpha_1\oplus\cdots\oplus \alpha_m)s_{XY}=(\alpha_1\oplus\cdots\oplus \alpha_m)\langle X|E_1\rangle$.

Let $\sigma$ map recursion variable $X$ in $E_1$ to $\langle X|E_1\rangle$, and let $\psi$ map recursion variable $Z_{XY}$ in $E$ to $s_{XY}$. It is sufficient to prove $s_{XY}=\psi(t_{XY})$ for recursion variables $Z_{XY}$ in $E$. Either $XY\equiv X_1Y_1$ or $XY\nequiv X_1Y_1$, we all can get $s_{XY}=\psi(t_{XY})$. So, $s_{XY}=\langle Z_{XY}|E\rangle$ for recursive variables $Z_{XY}$ in $E$ is a solution for $E$. Then by $RSP$, particularly, $\langle X_1|E_1\rangle=\langle Z_{X_1Y_1}|E\rangle$. Similarly, we can obtain $\langle Y_1|E_2\rangle=\langle Z_{X_1Y_1}|E\rangle$. Finally, $\langle X_1|E_1\rangle=\langle Z_{X_1Y_1}|E\rangle=\langle Y_1|E_2\rangle$, as desired.
\end{proof}

\subsection{Abstraction}

The unary abstraction operator $\tau_I$ ($I\subseteq \mathsf{Act}$) renames all atomic events in $I$ into $\tau$. APT with silent step and abstraction operator is called $APT_{\tau}$. 

\subsubsection{Operational Semantics}

The transition rules of operator $\tau_I$ are shown in Table \ref{TRForAbstraction}.

\begin{center}
    \begin{table}
        \[\frac{\langle x,s\rangle\xrightarrow{\alpha}\langle\surd,s'\rangle}{\langle\tau_I(x),s\rangle\xrightarrow{\alpha}\langle\surd,s'\rangle}\quad \alpha\notin I
        \quad\quad\frac{\langle x,s\rangle\xrightarrow{\alpha}\langle x',s'\rangle}{\langle\tau_I(x),s\rangle\xrightarrow{\alpha}\langle\tau_I(x'),s'\rangle}\quad \alpha\notin I\]

        \[\frac{\langle x,s\rangle\xrightarrow{\alpha}\langle\surd,s'\rangle}{\langle\tau_I(x),s\rangle\xrightarrow{\tau}\langle\surd,\tau(s)\rangle}\quad \alpha\in I
        \quad\quad\frac{\langle x,s\rangle\xrightarrow{\alpha}\langle x',s'\rangle}{\langle\tau_I(x),s\rangle\xrightarrow{\tau}\langle\tau_I(x'),\tau(s)\rangle}\quad \alpha\in I\]
        \caption{Transition rule of the abstraction operator}
        \label{TRForAbstraction}
    \end{table}
\end{center}

\begin{theorem}[Conservativity of $APT_{\tau}$ with guarded linear recursion]
$APT_{\tau}$ with guarded linear recursion is a conservative extension of APT with silent step and guarded linear recursion.
\end{theorem}

\begin{proof}
Since the transition rules of APT with silent step and guarded linear recursion are source-dependent, and the transition rules for abstraction operator in Table \ref{TRForAbstraction} contain only a fresh operator $\tau_I$ in their source, so the transition rules of $APT_{\tau}$ with guarded linear recursion is a conservative extension of those of APT with silent step and guarded linear recursion.
\end{proof}

\begin{theorem}[Congruence theorem of $APT_{\tau}$ with guarded linear recursion]
Rooted branching pomset bisimulation equivalence $\pbrb$ is a congruence with respect to $APT_{\tau}$ with guarded linear recursion.
\end{theorem}

\begin{proof}
Since the PTSS of $APT_{\tau}$ with guarded linear recursion in Table \ref{TRForAbstraction} is positive after reduction and in RBB cool format, rooted branching pomset bisimulation equivalence $\pbrb$ is a congruence with respect to $APT_{\tau}$ with guarded linear recursion.
\end{proof}

\subsubsection{Axiomatic System}

We design the axioms for the abstraction operator $\tau_I$ in Table \ref{AxiomsForAbstraction}.

\begin{center}
    \begin{longtable}{|c|c|}
      \caption{Axioms of abstraction}\\
      \hline No. &Axioms\\
      \hline
      \endfirsthead
      \multicolumn{2}{c}{\bfseries continuing: Axioms of abstraction} \\
      \hline No. &Axioms\\
      \hline
      \endhead
      \hline
      \multicolumn{2}{r}{to be continued\ldots} \\
      \endfoot
      \hline
      \multicolumn{2}{r}{end of Axioms of abstraction} \\
      \endlastfoot
      $TI1$ & $\alpha\notin I\quad \tau_I(\alpha)=\alpha$\\
      $TI2$ & $\alpha\in I\quad \tau_I(\alpha)=\tau$\\
      $TI3$ & $\tau_I(\mathbf{0})=\mathbf{0}$\\
      $TI4$ & $\tau_I(\mathbf{I})=\mathbf{I}$\\
      $TI5$ & $\tau_I(x+y)=\tau_I(x)+\tau_I(y)$\\
      $TI6$ & $\tau_I(x\circ y)=\tau_I(x)\circ\tau_I(y)$\\
      $TI7$ & $\tau_I(x\oplus y)=\tau_I(x)\oplus\tau_I(y)$\\
      $G22$ & $\phi\notin I\quad \tau_I(\phi)=\phi$\\
      $G23$ & $\phi\in I\quad \tau_I(\phi)=\tau_{\phi}$
      \label{AxiomsForAbstraction}
   \end{longtable}
\end{center}

\subsubsection{Soundness and Completeness}

\begin{theorem}[Soundness of $APT_{\tau}$ with guarded linear recursion]\label{SAPTABS}
Let $x$ and $y$ be $APT_{\tau}$ with guarded linear recursion terms. If $APT_{\tau}$ with guarded linear recursion $\vdash x=y$, then $x\pbrb y$.
\end{theorem}

\begin{proof}
Since rooted branching pomset bisimulation $\pbrb$ is both an equivalent and a congruent relation with respect to $APT_{\tau}$ with guarded linear recursion, we only need to check if each axiom in Table \ref{AxiomsForAbstraction} is sound modulo rooted branching pomset bisimulation equivalence. The proof is trivial and we omit it.
\end{proof}

Though $\tau$-loops are prohibited in guarded linear recursive specifications (see Definition \ref{GLRS}) in a specifiable way, they can be constructed using the abstraction operator, for example, there exist $\tau$-loops in the process term $\tau_{\{\alpha\}}(\langle X|X=aX\rangle)$. To avoid $\tau$-loops caused by $\tau_I$ and ensure fairness, the concept of cluster and $CFAR$ (Cluster Fair Abstraction Rule) \cite{CFAR} are still valid, we introduce them below.

\begin{definition}[Cluster]\label{CLUSTER}
Let $E$ be a guarded linear recursive specification, and $I\subseteq \mathsf{Act}$. Two recursion variable $X$ and $Y$ in $E$ are in the same cluster for $I$ iff there exist sequences of transitions $\langle X|E\rangle\xrightarrow{\mset{\beta_{11},\cdots, \beta_{1i}}}\cdots\xrightarrow{\mset{\beta_{m1},\cdots, \beta_{mi}}}\langle Y|E\rangle$ and $\langle Y|E\rangle\xrightarrow{\mset{\gamma_{11},\cdots, \gamma_{1j}}}\cdots\xrightarrow{\mset{\gamma_{n1},\cdots, \gamma_{nj}}}\langle X|E\rangle$, where $\beta_{11},\cdots,\beta_{mi},\gamma_{11},\cdots,\gamma_{nj}\in I\cup\{\tau\}$.

$\alpha_1\oplus\cdots\oplus \alpha_k$ or $(\alpha_1\oplus\cdots\oplus \alpha_k) X$ is an exit for the cluster $C$ iff: (1) $\alpha_1\oplus\cdots\oplus \alpha_k$ or $(\alpha_1\oplus\cdots\oplus \alpha_k) X$ is a summand at the right-hand side of the recursive equation for a recursion variable in $C$, and (2) in the case of $(\alpha_1\oplus\cdots\oplus \alpha_k) X$, either $\alpha_{\ell}\notin I\cup\{\tau\}({\ell}\in\{1,2,\cdots,k\})$ or $X\notin C$.
\end{definition}

\begin{center}
    \begin{longtable}{|c|c|}
      \caption{Cluster fair abstraction rule}\\
      \hline No. &Axioms\\
      \hline
      \endfirsthead
      \multicolumn{2}{c}{\bfseries continuing: Cluster fair abstraction rule} \\
      \hline No. &Axioms\\
      \hline
      \endhead
      \hline
      \multicolumn{2}{r}{to be continued\ldots} \\
      \endfoot
      \hline
      \multicolumn{2}{r}{end of Cluster fair abstraction rule} \\
      \endlastfoot
      $CFAR$ & If $X$ is in a cluster for $I$ with exits \\
           & $\{(\alpha_{11}\oplus\cdots\oplus \alpha_{1i})Y_1,\cdots,(\alpha_{m1}\oplus\cdots\oplus \alpha_{mi})Y_m, \beta_{11}\oplus\cdots\oplus \beta_{1j},\cdots,\beta_{n1}\oplus\cdots\oplus \beta_{nj}\}$, \\
           & then $\tau\circ\tau_I(\langle X|E\rangle)=$\\
           & $\tau\circ\tau_I((\alpha_{11}\oplus\cdots\oplus \alpha_{1i})\langle Y_1|E\rangle+\cdots+(\alpha_{m1}\oplus\cdots\oplus \alpha_{mi})\langle Y_m|E\rangle+\beta_{11}\oplus\cdots\oplus \beta_{1j}+\cdots+\beta_{n1}\oplus\cdots\oplus \beta_{nj})$
      \label{CFAR}
   \end{longtable}
\end{center}

\begin{theorem}[Soundness of $CFAR$]\label{SCFAR}
$CFAR$ is sound modulo rooted branching pomset bisimulation equivalence $\pbrb$.
\end{theorem}

\begin{proof}
Let $X$ be in a cluster for $I$ with exits $\{(\alpha_{11}\oplus\cdots\oplus \alpha_{1i})Y_1,\cdots,(\alpha_{m1}\oplus\cdots\oplus \alpha_{mi})Y_m,\beta_{11}\oplus\cdots\oplus \beta_{1j},\cdots,\beta_{n1}\oplus\cdots\oplus \beta_{nj}\}$. Then $\langle X|E\rangle$ can execute a string of atomic events from $I\cup\{\tau\}$ inside the cluster of $X$, followed by an exit $(\alpha_{i'1}\oplus\cdots\oplus \alpha_{i'i})Y_{i'}$ for $i'\in\{1,\cdots,m\}$ or $\beta_{j'1}\oplus\cdots\oplus \beta_{j'j}$ for $j'\in\{1,\cdots,n\}$. Hence, $\tau_I(\langle X|E\rangle)$ can execute a string of $\tau^*$ inside the cluster of $X$, followed by an exit $\tau_I((\alpha_{i'1}\oplus\cdots\oplus \alpha_{i'i})\langle Y_{i'}|E\rangle)$ for $i'\in\{1,\cdots,m\}$ or $\tau_I(\beta_{j'1}\oplus\cdots\oplus \beta_{j'j})$ for $j'\in\{1,\cdots,n\}$. And these $\tau^*$ are non-initial in $\tau\tau_I(\langle X|E\rangle)$, so they are truly silent by the axiom $B1$, we obtain $\tau\tau_I(\langle X|E\rangle)\approx_{rbs}\tau\circ\tau_I((\alpha_{11}\oplus\cdots\oplus \alpha_{1i})\langle Y_1|E\rangle+\cdots+(\alpha_{m1}\oplus\cdots\oplus \alpha_{mi})\langle Y_m|E\rangle+\beta_{11}\oplus\cdots\oplus \beta_{1j}+\cdots+\beta_{n1}\oplus\cdots\oplus \beta_{nj})$, as desired.
\end{proof}

\begin{theorem}[Completeness of $APT_{\tau}$ with guarded linear recursion and $CFAR$]\label{CCFAR}
Let $p$ and $q$ be closed $APT_{\tau}$ with guarded linear recursion and $CFAR$ terms, then, if $p\pbrb q$ then $p=q$.
\end{theorem}

\begin{proof}
Firstly, in the proof the Theorem \ref{CAPTTAU}, we know that each process term $p$ in APT with silent step and guarded linear recursion is equal to a process term $\langle X_1|E\rangle$ with $E$ a guarded linear recursive specification. And we prove if $\langle X_1|E_1\rangle\approx_{rbs}\langle Y_1|E_2\rangle$, then $\langle X_1|E_1\rangle=\langle Y_1|E_2\rangle$

The only new case is $p\equiv\tau_I(q)$. Let $q=\langle X|E\rangle$ with $E$ a guarded linear recursive specification, so $p=\tau_I(\langle X|E\rangle)$. Then the collection of recursive variables in $E$ can be divided into its clusters $C_1,\cdots,C_N$ for $I$. Let

\[(\alpha_{1i1}\oplus\cdots\oplus \alpha_{k_{i1}i1}) Y_{i1}+\cdots+(\alpha_{1im_i}\oplus\cdots\oplus \alpha_{k_{im_i}im_i}) Y_{im_i}+\beta_{1i1}\oplus\cdots\oplus \beta_{{\ell}_{i1}i1}+\cdots+\beta_{1im_i}\oplus\cdots\oplus \beta_{{\ell}_{im_i}im_i}\]

be the conflict composition of exits for the cluster $C_i$, with $i\in\{1,\cdots,N\}$.

For $Z\in C_i$ with $i\in\{1,\cdots,N\}$, we define

$s_Z\triangleq (\hat{\alpha_{1i1}}\oplus\cdots\oplus \hat{\alpha_{k_{i1}i1}}) \tau_I(\langle Y_{i1}|E\rangle)+\cdots+(\hat{\alpha_{1im_i}}\oplus\cdots\oplus \hat{\alpha_{k_{im_i}im_i}}) \tau_I(\langle Y_{im_i}|E\rangle)+\hat{\beta_{1i1}}\oplus\cdots\oplus \hat{\beta_{{\ell}_{i1}i1}}+\cdots+\hat{\beta_{1im_i}}\oplus\cdots\oplus \hat{\beta_{{\ell}_{im_i}im_i}}$

For $Z\in C_i$ and $\alpha_1,\cdots,\alpha_j\in \mathsf{Act}\cup\{\tau\}$ with $j\in\mathbb{N}$, we have

$(\alpha_1\oplus\cdots\oplus \alpha_j)\tau_I(\langle Z|E\rangle)$

$=(\alpha_1\oplus\cdots\oplus \alpha_j)\tau_I((\alpha_{1i1}\oplus\cdots\oplus \alpha_{k_{i1}i1}) \langle Y_{i1}|E\rangle+\cdots+(\alpha_{1im_i}\oplus\cdots\oplus \alpha_{k_{im_i}im_i}) \langle Y_{im_i}|E\rangle+\beta_{1i1}\oplus\cdots\oplus \beta_{{\ell}_{i1}i1}+\cdots+\beta_{1im_i}\oplus\cdots\oplus \beta_{{\ell}_{im_i}im_i})$

$=(\alpha_1\oplus\cdots\oplus \alpha_j)s_Z$

Let the linear recursive specification $F$ contain the same recursive variables as $E$, for $Z\in C_i$, $F$ contains the following recursive equation

$Z=(\hat{\alpha_{1i1}}\oplus\cdots\oplus \hat{\alpha_{k_{i1}i1}}) Y_{i1}+\cdots+(\hat{\alpha_{1im_i}}\oplus\cdots\oplus \hat{\alpha_{k_{im_i}im_i}})  Y_{im_i}+\hat{\beta_{1i1}}\oplus\cdots\oplus \hat{\beta_{{\ell}_{i1}i1}}+\cdots+\hat{\beta_{1im_i}}\oplus\cdots\oplus \hat{\beta_{{\ell}_{im_i}im_i}}$

It is easy to see that there is no sequence of one or more $\tau$-transitions from $\langle Z|F\rangle$ to itself, so $F$ is guarded.

For

$s_Z=(\hat{\alpha_{1i1}}\oplus\cdots\oplus \hat{\alpha_{k_{i1}i1}}) Y_{i1}+\cdots+(\hat{\alpha_{1im_i}}\oplus\cdots\oplus \hat{\alpha_{k_{im_i}im_i}}) Y_{im_i}+\hat{\beta_{1i1}}\oplus\cdots\oplus \hat{\beta_{{\ell}_{i1}i1}}+\cdots+\hat{\beta_{1im_i}}\oplus\cdots\oplus \hat{\beta_{{\ell}_{im_i}im_i}}$

is a solution for $F$. So, $(\alpha_1\oplus\cdots\oplus \alpha_j)\tau_I(\langle Z|E\rangle)=(\alpha_1\oplus\cdots\oplus \alpha_j)s_Z=(\alpha_1\oplus\cdots\oplus \alpha_j)\langle Z|F\rangle$.

So,

$\langle Z|F\rangle=(\hat{\alpha_{1i1}}\oplus\cdots\oplus \hat{\alpha_{k_{i1}i1}}) \langle Y_{i1}|F\rangle+\cdots+(\hat{\alpha_{1im_i}}\oplus\cdots\oplus \hat{\alpha_{k_{im_i}im_i}}) \langle Y_{im_i}|F\rangle+\hat{\beta_{1i1}}\oplus\cdots\oplus \hat{\beta_{{\ell}_{i1}i1}}+\cdots+\hat{\beta_{1im_i}}\oplus\cdots\oplus \hat{\beta_{{\ell}_{im_i}im_i}}$

Hence, $\tau_I(\langle X|E\rangle=\langle Z|F\rangle)$, as desired.
\end{proof}

\newpage\section{Application}\label{app} 

This chapter applies APT to three Transformer constructions from The Transformer Cookbook \cite{TTCB}: Most-Recent Induction, Dyck-1 recognition, and Dyck-1-2 recognition.
Each application first describes the construction and establishes its functional specification, then represents its components using the existing atomic constructions, and finally derives its external input-output behavior from recursive state-transition equations by abstraction. Functional correctness is established from the concrete computation before internal actions are hidden.

\subsection{Verification of a Transformer for Most-Recent Induction}

The subsection apply APT to the Most-Recent Induction construction presented in the Transformer Cookbook\cite{TTCB}. At each position, the construction predicts the symbol following the most recent earlier occurrence of the current symbol. If no earlier occurrence exists, it predicts the current symbol.

We first describe the construction and its functional specification. We then represent initialization, predecessor retrieval, induction retrieval, and unembedding using existing APT atomic constructions. Finally, we connect these components through recursive process equations and derive the expected external input-output behavior by abstraction. 

\subsubsection{A Transformer Construction of the Most-Recent Induction}\label{Most-recent-Cons}

Let $\Sigma$ be a finite nonempty alphabet, and let
\[
  x=x_1\cdots x_n\in\Sigma^+,\qquad n=|x|,
\]
be a nonempty input sequence, where $x_i$ denotes the symbol at position $i$, where $1\le i\le n$. Each symbol $\sigma\in\Sigma$ is represented by its one-hot encoding $e_\sigma$.
\[
\mathrm{WE}(x)=(e_{x_1},\ldots,e_{x_n}).
\]
For the example ,
\[ \mathrm{WE}(\mathtt{ACABDACDCA}) = (e_A,e_C,e_A,e_B,e_D,e_A,e_C,e_D,e_C,e_A).\]
For every symbol, use the predecessor construction to retrieve the
embedding of the predecessor symbol (returning 0 when none exists).
\[
\mathrm{Pred}(x_i)=
\begin{cases}
0,&i=1,\\
e_{x_{i-1}},&2\le i\le n.
\end{cases}
\]
Thus,
\[ \mathrm{Pred}(\mathtt{ACABDACDCA})=
(0,e_A,e_C,e_A,e_B,e_D,e_A,e_C,e_D,e_C).\]

At position $i$, attention compares $e_{x_i}$ with the predecessor embeddings at candidate positions $j\le i$. A position $j$ matches when
\[\mathrm{Pred}(x_j)=e_{x_i}.\]
For such a position, $x_{j-1}=x_i$, and the value $e_{x_j}$ represents the symbol following that earlier occurrence. Rightmost selection chooses the largest matching position index and returns its value:
\[
\mathrm{Retrieve}(x)_i=
\begin{cases}
e_{x_j}, &\text{if there exists }j\text{ with }2\le j\le i \text{ and }x_{j-1}=x_i,\\[2pt]
e_{x_i},&\text{otherwise}.
\end{cases}
\]
Here, $i$ is the current query position, and $j$ denotes the largest index satisfying $2\le j\le i$ and $x_{j-1}=x_i$, whenever such an index exists. The largest candidate index corresponds to the most recent earlier occurrence.

For the running example,
\[ \mathrm{Retrieve}(\mathtt{ACABDACDCA}) = (e_A,e_C,e_C,e_B,e_D,e_B,e_A,e_A,e_D,e_C). \]

Unembedding decodes the retrieved embeddings:
\[
\mathrm{Induct}_{\mathrm{rightmost}}(x_i)=
\begin{cases}
x_j,&\text{if a matching candidate exists},\\
x_i,&\text{otherwise},
\end{cases}
\]
with j defined above.

Thus,
\[ \mathrm{Induct}_{\mathrm{rightmost}}(\mathtt{ACABDACDCA}) = \mathtt{ACCBDBAADC}.\]

\subsubsection{Algebraic Representation}\label{Most-recent-rep}

\runhead{Input interface.} The action 
\[R(x),\quad x\in\Sigma^+\]
receives the input sequence $x = x_1\cdot x_n$ and stores it. It performs no Transformer computation. Sequential composition transfers control to initialization.

The input and output actions describe interaction with the environment; they are not additional Transformer computational primitives.
\runhead{Embedding.} 
The initialization component is represented by
\[(\alpha_1\oplus\alpha_2)\circ\alpha_5. \]
Here:
\begin{itemize}
  \item $\alpha_1$ is the word-embedding action. It supplies the one-hot current-symbol embeddings, computes $\mathrm{WE}(x_i)=e_{x_i}$ at each position.
  \item $\alpha_2$ is the position-embedding action. It supplies $i/n$ for predecessor retrieval in the first self-attention component and rightmost selection in the second self-attention component.
  \item $\alpha_5$ is the addition action, which merges the outputs of $\alpha_1$ and $\alpha_2$ coordinatewise.
\end{itemize}
The two embedding actions execute in parallel, followed by their coordinatewise addition.

\runhead{First Transformer layer.} 
The first transformer component uses the predecessor construction to retrieve the embedding of the immediately preceding symbol. It is represented by
\[ (\alpha_{12}\oplus\alpha_{16})\circ\alpha_5. \]
Here:
\begin{itemize}
  \item $\alpha_{12}$ is the sequence identity action. It preserves the current-symbol embeddings and positional information.
  \item $\alpha_{16}$ is the predecessor action. It retrieves $\mathrm{Pred}(x_i)$ and routes the result to coordinates separate from the preserved embeddings.
  \item $\alpha_5$ combines the preserved information and predecessor embeddings.
\end{itemize}
The parallel composition preserves the current information while computing the predecessor embedding. The subsequent addition retains both results for the second Transformer layer.

\runhead{Second Transformer layer.} 
The second Transformer layer retrieves the symbol following the most recent earlier occurrence of the current symbol. Its self-attention component and residual connection are represented by
\[
(\alpha_{12}\oplus\alpha_{18})\circ\alpha_5.
\]
Here:
\begin{itemize}
  \item $\alpha_{12}$ is the sequence identity action. It preserves the information produced by the first Transformer layer.
  \item $\alpha_{18}$ is the multi-head action, configured with one effective retrieval head. It uses $\mathrm{WE}(x_i)$ as the query content, $\mathrm{Pred}(x_j)$ as the key content, and $\mathrm{WE}(x_j)$ as the value at each candidate position $j\le i$. Content matching and rightmost selection produce $\mathrm{Retrieve}(x_i)$. If no match exists, the configuration returns $\mathrm{WE}(x_i)$.
  \item $\alpha_5$ adds the attention output to the preserved input, retaining the retrieved embedding for unembedding.
\end{itemize}
The second-layer FFN is configured as zero because the attention component already produces the required prediction embedding. This inactive FFN is omitted from the process representation.

\runhead{Unembedding.} 
The unembedding action 
\[\alpha_{19}\]
decodes each retrieved one-hot embedding $\mathrm{Retrieve}(x_i)$ into its corresponding symbol. Applied at every position, it produces the prediction sequence
\[
\mathrm{Induct}_{\mathrm{rightmost}}(x).
\]

\runhead{Output interface.} 
The output action reads and emits the computed prediction sequence. Its occurrence is therefore labelled
\[
W\big(\mathrm{Induct}_{\mathrm{rightmost}}(x)\big)
\]

\subsubsection{Verification of the Most-Recent Induction Transformer}\label{Most-recent-V}

We now verify that the Most-Recent Induction Transformer exhibits the expected external behavior: it receives an input sequence, computes its Most-Recent Induction predictions, emits the result, and returns to the input-waiting process.
\runhead{State-transition equations.} 
Using the component representations in Section~\ref{Most-recent-rep}, define
\[
\begin{aligned}
T &= R(x)\circ T_2,\\
T_2 &= \big((\alpha_1\oplus\alpha_2)\circ\alpha_5\big)\circ T_3,\\
T_3 &= \big((\alpha_{12}\oplus\alpha_{16})\circ\alpha_5\big)\circ T_4,\\
T_4 &= \big((\alpha_{12}\oplus\alpha_{18})\circ\alpha_5\big)\circ T_5,\\
T_5 &= \alpha_{19}\circ T_6,\\
T_6 &= W\big(\mathrm{Induct}_{\mathrm{rightmost}}(x)\big)\circ T.
\end{aligned}
\]

\runhead{Abstraction and external behavior.} 
Use the abstraction set $I$ defined earlier in the manuscript. 
\[
I\subseteq \mathsf{Act}\cup \mathsf{Grd} 
\]
It contains all configured internal actions occurring in the process equations, while the input and output actions remain visible.
\begin{theorem}[Correctness of the Most-Recent Induction Transformer]
  For every received input sequence $x$, the transformer construction of Most-Recent Induction computes and emits its Most-Recent Induction prediction sequence and returns to input-waiting state. Thus, $\tau_I(T)$ satisfies the external behavior equation:
  \[
    \tau_I(T)= R(x)\circ W\big(\mathrm{Induct}_{\mathrm{rightmost}}(x)\big)\circ\tau_I(T).
  \]

  \begin{proof}
    Section~\ref{Most-recent-Cons} establishes that the predecessor and retrieval computations return the required prediction embedding at each position, including the no-match case. Unembedding therefore produces the stated output sequence.
    
    Unfolding the state-transition equations and applying TI1, TI2, TI6, and TI7 gives
    \[
    \begin{aligned}
    \tau_I(T) &= R(x)\circ\tau_I(T_2),\\
    \tau_I(T_2) &= (\tau\oplus\tau)\circ\tau\circ\tau_I(T_3),\\
    \tau_I(T_3) &= (\tau\oplus\tau)\circ\tau\circ\tau_I(T_4),\\
    \tau_I(T_4) &= (\tau\oplus\tau)\circ\tau\circ\tau_I(T_5),\\
    \tau_I(T_5) &= \tau\circ\tau_I(T_6),\\
    \tau_I(T_6) &= W\big(\mathrm{Induct}_{\mathrm{rightmost}}(x)\big)\circ\tau_I(T).
    \end{aligned}
    \]
    By B4, \(\tau\oplus\tau=\tau\), together with A5, this gives
    \[
    \begin{aligned}
    \tau_I(T) &= R(x)\circ\tau_I(T_2),\\
    \tau_I(T_2) &= \tau\circ\tau\circ\tau_I(T_3),\\
    \tau_I(T_3) &= \tau\circ\tau\circ\tau_I(T_4),\\
    \tau_I(T_4) &= \tau\circ\tau\circ\tau_I(T_5),\\
    \tau_I(T_5) &= \tau\circ\tau_I(T_6),\\
    \tau_I(T_6) &= W\bigl(\mathrm{Induct}_{\mathrm{rightmost}}(x)\bigr)\circ\tau_I(T).
    \end{aligned}
    \]
    For each input branch, substitution gives seven sequential silent actions between input and output.
    \[
    \begin{aligned}
    \tau_I(T) = R(x) &\circ\tau\circ\tau\circ\tau\circ\tau\circ\tau\circ\tau\circ\tau\\
    &\circ W\big(\mathrm{Induct}_{\mathrm{rightmost}}(x)\big)\circ\tau_I(T).\end{aligned}
    \]
    For any continuation $Y$, A5 and B1 give
    \[
    R(x)\circ(\tau\circ Y)
    \overset{A5}{=}(R(x)\circ\tau)\circ Y
    \overset{B1}{=}R(x)\circ Y.
    \]
    Applying this identity repeatedly within each input branch removes the seven sequential silent actions. Hence,
    \[
      \tau_I(T) = R(x)\circ W\big(\mathrm{Induct}_{\mathrm{rightmost}}(x)\big)\circ\tau_I(T).
    \]
    The resulting equation exhibits the required external behavior: the transformer receives an input sequence, computes its Most-Recent Induction predictions, emits the result, and returns to the input-waiting process.
  \end{proof}
\end{theorem}

\subsection{Verification of a Transformer for Dyck-1}

In this subsection, we apply APT to the verification of the uniform two-layer Transformer for the Dyck-1 language presented in the Transformer Cookbook\cite{TTCB}. The construction considered here is uniform: its parameters are independent of the input length, and the transformer computes the quantities required for recognition rather than an internal binary acceptance bit.

The verification has three complementary parts. First, we describe the numerical construction and establish its recognition criterion. Second, we represent its embedding, attention, and feed-forward components using existing atomic constructions. Finally, we connect these components through recursive equations and derive the external input-output behavior by abstraction. 

\subsubsection{A Transformer Construction for Dyck-1}

The Dyck-1 language consists of well-nested strings over the alphabet
\[ \Sigma=\{\texttt{(},\texttt{)}\}. \] 
A string \(w\in\Sigma^*\) belongs to Dyck-1 exactly when it satisfies the following two conditions:
\begin{itemize}
  \item Property 1 (Prefix Condition): In every prefix of the string, the count of open brackets must be greater than or equal to the count of close brackets.
  \item Property 2 (Balance Condition): The total count of open and close brackets in the entire string must be equal.
\end{itemize}
The construction below is stated for nonempty inputs because it inspects the representation at the final position. The empty string may be accepted separately by convention. 
Let $w = w_1\cdot w_n\in\Sigma^+$ be a nonempty input string. Encode its symbols by
\[
  o_i=
  \begin{cases}
    1, & w_i=\texttt{(},\\
    -1, & w_i=\texttt{)}.
  \end{cases}
\]

For each \(i\in[n]\), let \(O_i\) and \(C_i\) be the numbers of opening and closing parentheses, respectively, in the prefix \(w_1\cdots w_i\). The corresponding prefix balance is
\[
  B_i = O_i-C_i = \sum_{j=1}^{i}o_j.
\]

The embedding initializes
\[
(o_i,0,0,0)^{\mathsf T}.
\]

The first attention component computes the normalized prefix balance:
\[
\frac1i\sum_{j=1}^{i}o_j=\frac{B_i}{i}.
\]

The first FFN computes the prefix-error signal
\[
E_i=\operatorname{ReLU}\left(-\frac{B_i}{i}\right).
\]

The second attention component computes its prefix average:
\[
t_i=\frac1i\sum_{j=1}^{i}E_j.
\]

The final residual representation at position \(i\) is
\[
  z_i = 
  \begin{bmatrix} z_{i,1}\\ z_{i,2}\\ z_{i,3}\\ z_{i,4} \end{bmatrix}
  =
  \begin{bmatrix} o_i \\ B_i/i \\ E_i \\ t_i \end{bmatrix}.
\]

The representation therefore evolves as
\[
  \begin{bmatrix} o_i \\ 0 \\ 0 \\ 0 \end{bmatrix}
  \xrightarrow{\text{Compute Balance}}
  \begin{bmatrix} o_i \\ B_i/i \\ 0 \\ 0 \end{bmatrix}
  \xrightarrow{\text{Negative } B_i \text{ Error?}}
  \begin{bmatrix} o_i \\ B_i/i \\ E_i \\ 0 \end{bmatrix}
  \xrightarrow{\text{Sum of Errors}}
  \begin{bmatrix} o_i \\ B_i/i \\ E_i \\ t_i \end{bmatrix}
\]

Write the final representation as
\[
z_i=(o_i,B_i/i,E_i,t_i)^{\mathsf T}.
\]

At the final position,
\[
  z_{n,2}=\frac{B_n}{n},\qquad z_{n,4}=t_n.
\]

Since \(n>0\),
\[
z_{n,2}=\frac{B_n}{n}=0
\iff
B_n=0,
\]
so \(z_{n,2}=0\) is equivalent to the balance condition. Moreover,
\[
E_i=0\iff B_i\ge0.
\]
and all error signals $E_i$ are nonnegative. Therefore,
\[
t_n=0
\iff
E_i=0\text{ for every }i\in[n]
\iff
B_i\geq0\text{ for every }i\in[n].
\]

Thus, \(z_{n,4}=0\) is equivalent to the prefix condition. It follows that
\[
w\in\mathrm{Dyck}\text{-}1
\iff
z_{n,2}=0\land z_{n,4}=0.
\]

For later use, define the expected numerical output for $w$ by
\[
  o_D := \left( \frac{B_n(w)}{n}, t_n(w) \right).
\]
Then
\[
  w\in\mathrm{Dyck}\text{-}1
  \quad\Longleftrightarrow\quad
  o_D=(0,0).
\]
The equality test is an external observation. It is not an internal comparison layer or Boolean classifier.

\subsubsection{Algebraic Representation}\label{Dyck-1-rep}

The construction uses causal uniform averaging, residual connections, and a four-dimensional residual workspace. No positional embedding or layer normalization is required. 

Repeated occurrences of the same action symbol may therefore denote differently configured instances of the same construction. For example, the two occurrences of $\alpha_{13}$ below use different source and target workspace coordinates. All such parameter choices are meta-level configurations only and introduce neither new atomic actions nor new operators into APT.

\runhead{Input interface.}The input action
\[ R(w) \]
receives the input string $w$ from the external environment and stores it. This interface action does not itself perform a Transformer computation. Once it has terminated, sequential composition transfers control to embedding.

\runhead{Embedding.} 
The first internal Transformer action is the existing word-embedding action 
\[\alpha_1\]
which reads the symbols $w_i$ from the input component, yields the first component $o_i$ at each position, and initializes the residual representation at position $i$ as
\[
\begin{bmatrix}
o_i\\ 0\\ 0\\ 0
\end{bmatrix}.
\]
No separate position-embedding action $\alpha_2$ is required in this construction.

\runhead{First self-attention component.} 
The first self-attention component represented by \[\left(\alpha_{12}\oplus\alpha_{13}\right)\circ\alpha_5.\]
Here:
\begin{itemize}
  \item $\alpha_{12}$ is the sequence-identity construction used to preserve input information;
  \item $\alpha_{13}$ is the uniform-average action, which reads the first coordinate and computes its uniform average over positions $1,\cdot ,i$. It writes the causal prefix average $B_i/i$ to the second workspace coordinate and produce $0$ in the other output coordinates;
  \item $\alpha_5$ merges the preserved input representation with the newly attention output. 
\end{itemize}
Consequently, this component performs
\[
\begin{bmatrix}
o_i\\ 0\\ 0\\ 0
\end{bmatrix}
\longmapsto
\begin{bmatrix}
o_i\\ B_i/i\\ 0\\ 0
\end{bmatrix}.
\]

\runhead{First feed-forward component.} 
The first feed-forward component and its residual connection are \[\left(\alpha_3\oplus\left(\alpha_6\circ\alpha_{10}\right)\right)\circ\alpha_5.\] Here:
\begin{itemize}
  \item $\alpha_3$ is the vector-level identity construction that preserves the residual representation while the FFN branch computes the new error signal.
  \item $\alpha_6$ reads the normalized prefix balance $B_i/i$ and multiplies it by $-1$;
  \item $\alpha_{10}$ is the continuous-piecewise-linear action, applies ReLU to produce the prefix error signal $E_i=\operatorname{ReLU}(-B_i/i)$ and routes the result to the third coordinate;
  \item $\alpha_5$ adds the computational output to the retained identity-branch output.
\end{itemize}
The computational branch produces
\[
\begin{bmatrix}
  0\\ 0\\ E_i\\ 0 
\end{bmatrix}
\]
Thus, this is the correction required to obtain
\[
\begin{bmatrix}
o_i\\ B_i/i\\ 0\\ 0
\end{bmatrix}
\longmapsto
\begin{bmatrix}
o_i\\ B_i/i\\ E_i\\ 0
\end{bmatrix}.
\]

\runhead{Second self-attention component} 
The second self-attention component is represented by \[\left(\alpha_{12}\oplus\alpha_{13}\right)\circ\alpha_5.\]
Here:
\begin{itemize}
  \item $\alpha_{12}$ preserves the residual sequence;
  \item $\alpha_{13}$ now reads the third coordinate (the error signals $E_i$) and computes its causal prefix average $t_i$. It writes this value into the fourth coordinate;
  \item $\alpha_5$ adds the routed output to the preserved representation. 
\end{itemize}
The resulting transformation is
\[
\begin{bmatrix}
o_i\\ B_i/i\\ E_i\\ 0
\end{bmatrix}
\longmapsto
\begin{bmatrix}
o_i\\ B_i/i\\ E_i\\ t_i
\end{bmatrix}.
\]

\runhead{Second feed-forward component.} 
No additional feed-forward computation is required after the second self-attention stage. In the Dyck-1 construction, the second attention head already writes the accumulated error $t_i$ into the final workspace coordinate, after which recognition is specified by inspecting the final-position values. We therefore omit the second-layer FFN.

\runhead{Output interface} 
The output action reads the second and fourth coordinates at the final position:
\[
(z_{n,2},z_{n,4})
= \left(\frac{B_n(w)}n,t_n(w)\right)
= o_D.
\]

Its occurrence is therefore labelled 
\[
\mathsf W(o_D),
\]
which emits a two-dimensional numerical result without performing the external acceptance test.

\subsubsection{Verification of the Dyck-1 Transformer}\label{Dyck-1-V}

We now verify that the Dyck-1 Transformer exhibits the expected external behavior: it receives a nonempty input string, computes and emits the corresponding numerical result, and returns to its input-waiting process.

\runhead{State-transitions equations.} 
For each $w\in\Sigma^+$, define its processes are as follows:
\[
\begin{aligned}
T &= R(w)\circ T_2,\\
T_2 &= \alpha_1\circ T_3,\\
T_3 &= \left(\left(\alpha_{12}\oplus\alpha_{13}\right)\circ\alpha_5\right)\circ T_4,\\
T_4 &= \left(\left(\alpha_3\oplus\left(\alpha_6\circ\alpha_{10}\right)\right)\circ\alpha_5\right)\circ T_5,\\
T_5 &= \left(\left(\alpha_{12}\oplus\alpha_{13}\right)\circ\alpha_5\right)\circ T_6,\\
T_6 &= W(o_D)\circ T.
\end{aligned}
\]

\runhead{Abstraction of internal actions.} 
All configured internal actions belong to $I$. The interface actions $R(w)$ and $W(o_D)$ remain visible.

\begin{theorem}[Correctness of the Dyck-1 Transformer.]
For every received input string $w$, the transformer construction of Dyck-1 computes and emits its numerical result $o_D$ and returns to input-waiting state. 
Thus, $\tau_I(T)$ satisfies the external behavior equation:
\[
\tau_I(T) = R(w)\circ W(o_D)\circ\tau_I(T).
\]
\begin{proof}
  First, for the process calculation, axiom P1 gives
  \[\alpha_3\oplus(\alpha_6\circ\alpha_{10}) = (\alpha_3\oplus\alpha_6)\circ\alpha_{10}.\]
  Together with A5, this rewrites the feed-forward residual process as
  \[
    T_4 = (\alpha_3\oplus\alpha_6)\circ\alpha_{10}\circ\alpha_5\circ T_5.
  \]
  The abstraction axioms TI1, TI2, TI6, and TI7 apply to the residual-process equations, and using B4 to simplify 
  \[\tau\oplus\tau = \tau,\] 
  we obtain
  \[
  \begin{aligned}
  \tau_I(T) &= R(w)\circ\tau_I(T_2),\\
  \tau_I(T_2) &= \tau\circ\tau_I(T_3),\\
  \tau_I(T_3) &= \tau\circ\tau\circ\tau_I(T_4),\\
  \tau_I(T_4) &= \tau\circ\tau\circ\tau\circ\tau_I(T_5),\\
  \tau_I(T_5) &= \tau\circ\tau\circ\tau_I(T_6),\\
  \tau_I(T_6) &= W(o_D)\circ\tau_I(T).
  \end{aligned}
  \]
  Substitution yields eight sequential silent actions between the visible input and output actions.
  For any continuation $Y$, A5 and B1 give
  \[
  R(x)\circ(\tau\circ Y)
  \overset{A5}{=}(R(x)\circ\tau)\circ Y
  \overset{B1}{=}R(x)\circ Y.
  \]
  Applying this equality repeatedly within each input branch yields
  \[
  \tau_I(T) = R(w)\circ W(o_D)\circ\tau_I(T)
  \]
    
  The resulting equation exhibits the required external behavior: the transformer receives a string, emits its corresponding numerical result, and returns to the input-waiting process.
  \end{proof}
\end{theorem}
  
\subsection{Verification of a Transformer for Dyck-1-2}

In this subsection, we apply APT to the depth-two bracket-matching construction for the Dyck-1-2 presented in the Transformer Cookbook\cite{TTCB}. Unlike the running-balance construction for Dyck-1, this Transformer maintains an active bit at every position and performs two rounds of nearest-neighbor matching.

We first summarize the construction and its recognition criterion. We then represent its embedding, positional encoding, attention, and feed-forward components using existing APT atomic constructions. Finally, we connect these components through recursive process equations and derive the expected external input-output behavior by abstraction.

\subsubsection{A Transformer Construction of Dyck-1-2}\label{Dyck-1-2-Con}

The Dyck-1-2 language consists of well-nested strings with one type of bracket and maximum nesting depth two. Let $\Sigma=\{\mathtt{(},\mathtt{)}\}$, and let
\[
w=w_1\cdots w_n\in\Sigma^+,\qquad n=|w|,
\]
be a nonempty input string, where $w_i$ is the bracket at position $i$.

For \(1\le i\le n\), define
\[
o_i=
\begin{cases}
+1,&w_i=\texttt{(},\\
-1,&w_i=\texttt{)}.
\end{cases}
\]

We initialize our embedding vector $x_i$ at each position i as follows:
\[
  x_i=
  \begin{bmatrix}
  o_i\\ 1\\ 0\\ 0\\ i/n
  \end{bmatrix},
\]

The second coordinate is an "active bit". Its initial value is $1$, indicating that the bracket has not yet been matched. A successful match changes this coordinate to $0$. The third and fourth coordinates provide scratch space for intermediate computations, and the fifth stores positional information.

\runhead{Identifying depth-one matches.}
The first step retrieves the bracket types immediately to the left and right of each position using the predecessor and successor constructions with the corresponding strict masks. The intermediate representation contains the current bracket type and its two neighboring bracket types. An FFN then marks adjacent opening–closing pairs as matched and clears the scratch coordinates:
\[
\begin{bmatrix}
o_i\\ 1\\ 0\\ 0\\ i/n
\end{bmatrix}
\xrightarrow{A^{\mathrm{left}},\,A^{\mathrm{right}}}
\begin{bmatrix}
o_i\\ 1\\ o_{i-1}\\ o_{i+1}\\ i/n
\end{bmatrix}
\xrightarrow{\text{FFN: Mark Matched}}
\begin{bmatrix}
o_i\\ a_i\\ 0\\ 0\\ i/n
\end{bmatrix},
\]
where
\[
a_i=
\begin{cases}
0,&o_i=+1\text{ and }o_{i+1}=-1,\\
0,&o_i=-1\text{ and }o_{i-1}=+1,\\
1,&\text{otherwise}.
\end{cases}
\]

Thus, both brackets in each adjacent opening–closing pair are marked as matched. 

\runhead{Identifying depth-two matches.}
The second step searches for the nearest unmatched bracket on each side, ignoring positions whose active bits are zero. This selection through the UHAT expression 
\[
\operatorname{UHAT}\!\left(
\left[
\frac1n-(a_1-1),\,
\frac2n-(a_2-1),\,
\ldots,\,
\frac nn-(a_n-1)
\right]
\right)_i,
\qquad i=1,\ldots,n.
\]

The source describes obtaining the required quantities through a linear transformation of the intermediate representation, using the attention query $W_{(Q)}$ and key matrices $W_{K}$ with a bias term. We then proceed analogously to first step, the workspace changes as follows:
\[
\begin{bmatrix}
o_i\\ a_i\\ 0\\ 0\\ i/n
\end{bmatrix}
\xrightarrow{\text{Select Nearest Active}}
\begin{bmatrix}
o_i\\ a_i\\ l_i\\ r_i\\ i/n
\end{bmatrix}
\xrightarrow{\text{Check Match}}
\begin{bmatrix}
o_i\\ a'_i\\ 0\\ 0\\ i/n
\end{bmatrix}.
\]

Here, $l_i$ and $r_i$ denote the types of the nearest unmatched brackets to the left and right, with 0 denoting a missing neighbor. The nearest unmatched bracket types \(l_i\) and \(r_i\) are defined as above:
\[
l_i=
\begin{cases}
o_{\max\{j\in\{1,\ldots,i-1\}:a_j=1\}},
&\text{if }\{j\in\{1,\ldots,i-1\}:a_j=1\}\neq\varnothing,\\
0,&\text{otherwise},
\end{cases}
\]

\[
r_i=
\begin{cases}
o_{\min\{j\in\{i+1,\ldots,n\}:a_j=1\}},
&\text{if }\{j\in\{i+1,\ldots,n\}\neq\varnothing,\\
0,&\text{otherwise}.
\end{cases}
\]

The maximum selects the largest active position index smaller than $i$, while the minimum selects the smallest active position index greater than $i$. 

The FFN updates the active bit according to
\[
a'_i=
\begin{cases}
0,&a_i=0,\\
0,&a_i=1,\ o_i=+1,\ r_i=-1,\\
0,&a_i=1,\ o_i=-1,\ l_i=+1,\\
1,&\text{otherwise}.
\end{cases}
\]
After the second step, the string is accepted if and only if all active bits are 0.

\runhead{Recognition criterion.}
The final numerical output is the active-bit sequence
\[
(a'_1,\ldots,a'_n).
\]
The recognition condition is
\[
w\in\mathrm{Dyck}\text{-}1\text{-}2
\iff
a'_i=0\quad\text{for every }i=1,\ldots,n.
\]

The all-zero test is performed by the external environment.

\subsubsection{Algebraic Representation}\label{Dyck-1-2-rep}

\runhead{Input interface.} 
The input action is represented by
$R(w)$
which receives the input string $w$ defined in Section~\ref{Dyck-1-2-Con} and stores it. Sequential composition then transfers control to the embedding component.

\runhead{Embedding.} 
The initialization component is represented by
\[
(\alpha_1\oplus\alpha_2)\circ\alpha_5.
\]
Here:
\begin{itemize}
  \item $\alpha_1$ is the word embedding action, which supplies the bracket type $o_i$ and initializes the active bit $a_i$.
  \item $\alpha_2$ is the position embedding action used to supply $i/n$, which enables attention to identify neighboring positions and select the nearest unmatched brackets in the second matching step.
  \item $\alpha_5$ is the addition action, which adds the outputs of $\alpha_1$ and $\alpha_2$ coordinatewise to obtain the initial workspace vector $ x_i=(o_i,1,0,0,i/n)^{\mathsf T}$.
\end{itemize}

\runhead{First attention component.} The first attention component retrieves the bracket types immediately preceding and succeeding each position and stores them in the two scratch coordinates. Its algebraic representation is
\[
(\alpha_{12}\oplus\alpha_{18})\circ\alpha_5.
\]
Here:
\begin{itemize}
  \item $\alpha_{12}$ is the sequence identity action. It preserves the current residual sequence, including the bracket types, active bits, and positional encodings.
  \item $\alpha_{18}$ is a configured multi-head attention action containing two directional heads. The left head retrieves the bracket type $o_{i-1}$ at position $i-1$, while the right head retrieves $o_{i+1}$ at position $i+1$, using the corresponding strict masks. Their outputs are routed to the third and fourth coordinates, a missing neighbor contributes 0.
  \item $\alpha_5$ adds the preserved residual sequence and the attention output coordinatewise. 
\end{itemize}
Thus, the component performs
\[
\begin{bmatrix} 
o_i\\ 1\\ 0\\ 0\\ i/n
\end{bmatrix}
\longmapsto
\begin{bmatrix}
o_i\\ 1\\ o_{i-1}\\ o_{i+1}\\ i/n
\end{bmatrix}.
\]

\runhead{First feed-forward component.} 
The first feed-forward component identifies adjacent opening–closing pairs, changes their active bits to 0, and resets the third and fourth coordinates to 0. 
The active-bit update is 
\[
U_1 = (\alpha_3\oplus\alpha_{10})\circ\alpha_5.
\]
Here:
\begin{itemize}
  \item $\alpha_3$ is the vector identity action. It preserves the input vector at each position.
  \item $\alpha_{10}$ is a configured piecewise-linear action. It reads the current bracket type and the two retrieved neighbor types, determines $a_i$ according to the first-step matching rule, and routes it only to the second coordinate.
  \item $\alpha_5$ adds the correction to the preserved input vector. It changes the active bit from $1$ to $a_i$, resets the third and fourth coordinates to $0$, and preserves the bracket type and positional encoding.
\end{itemize}
The temporary coordinates are cleared by
\[
C_3 = (\alpha_3\oplus\alpha_6)\circ\alpha_5,
\]
\[
C_4=(\alpha_3\oplus\alpha_6)\circ\alpha_5.
\]
In $C_3$:
\begin{itemize}
  \item $\alpha_3$ preserves the representation by $U_1$;
  \item $\alpha_6$ reads the third coordinate, multiplies its value by $-1$, and routes the result to the third output coordinate, with all other output coordinates set to zero;
  \item $\alpha_5$ adds the separately retained branch outputs coordinatewise. The third coordinate becomes $o_{i-1}+(-o_{i-1})=0$, while the remaining coordinates are preserved.
\end{itemize}
In $C_4$:
\begin{itemize}
  \item $\alpha_3$ preserves the representation produced by $C_3$;
  \item $\alpha_6$ reads the fourth coordinate, multiplies its value by -1, and routes the result to the fourth output coordinate, with all other output coordinates set to zero.
  \item $\alpha_5$ adds the branch outputs coordinatewise. The fourth coordinate becomes $o_{i+1}+(-o_{i+1})=0$, while the remaining coordinates are preserved.
\end{itemize}

Thus, the complete feed-forward component
\[
\begin{aligned}
F_1
&=U_1\circ C_3\circ C_4 \\ &=\big((\alpha_3\oplus\alpha_{10})\circ\alpha_5\big)\circ \big((\alpha3\oplus\alpha_6)\circ\alpha_5\big)\circ \big((\alpha_3\oplus\alpha_6)\circ\alpha_5\big). 
\end{aligned}
\]
updates the active bit and then resets the third and fourth coordinates to zero.
It performs
\[
\begin{bmatrix}
o_i\\ 1\\ o_{i-1}\\ o_{i+1}\\ i/n
\end{bmatrix}
\longmapsto
\begin{bmatrix}
o_i\\ a_i\\ 0\\ 0\\ i/n
\end{bmatrix}.
\]

\runhead{Second attention component.} 
The second attention component retrieves the nearest unmatched bracket on each side whenever such a bracket exists. Its algebraic representation is
\[
(\alpha_{12}\oplus\alpha_{18})\circ\alpha_5.
\]
Here:
\begin{itemize}
  \item $\alpha_{12}$ preserves the residual sequence produced by the first matching step, including its updated active bits.
  \item $\alpha_{18}$ is a configured multi-head attention action containing two directional heads. The left and right heads retrieve $l_i$ and $r_i$, respectively, and route them to the third and fourth coordinates.
  \item $\alpha_5$ adds the attention output to the preserved residual sequence.
\end{itemize}
Whenever the corresponding active neighbors exist, this gives
\[
\begin{bmatrix}
o_i\\ a_i\\ 0\\ 0\\ i/n
\end{bmatrix}
\longmapsto
\begin{bmatrix}
o_i\\ a_i\\ l_i\\ r_i\\ i/n
\end{bmatrix}.
\]
\runhead{Second feed-forward component.} 
The second feed-forward component matches adjacent opening–closing pairs in the remaining active sequence, preserves previously matched positions, and resets the third and fourth coordinates to $0$. 
The Second active bit update is
\[
U_2 = (\alpha_3\oplus\alpha_{10})\circ\alpha_5.
\]
Here:
\begin{itemize}
  \item $\alpha_3$ preserves the input vector at each position.
  \item $\alpha_{10}$ reads $o_i,a_i,l_i,r_i$ and implements
    \[
    a'_i=
    \begin{cases}
    0,&a_i=0,\\
    0,&a_i=1,\ o_i=+1,\ r_i=-1,\\
    0,&a_i=1,\ o_i=-1,\ l_i=+1,\\
    1,&\text{otherwise}.
    \end{cases}
    \]
    This correction updates the active bit.
  \item $\alpha_5$ adds the correction to the preserved input vector, producing $a'_i$ in the second coordinate.
\end{itemize}
Using the same clearing components $C_3$ and $C_4$ as in the first feed-forward component, the second feed-forward component is represented by
\[
\begin{aligned}
F_2
&=U_2\circ C_3\circ C_4 \\ &=\big((\alpha_3\oplus\alpha_{10})\circ\alpha_5\big)\circ \big((\alpha_3\oplus\alpha_6)\circ\alpha_5\big)\circ \big((\alpha_3\oplus\alpha_6)\circ\alpha_5\big)
\end{aligned}
\]
The resulting transformation is
\[
\begin{bmatrix}
o_i\\ a_i\\ l_i\\ r_i\\ i/n
\end{bmatrix}
\longmapsto
\begin{bmatrix}
o_i\\ a'_i\\ 0\\ 0\\ i/n
\end{bmatrix}.
\]

\runhead{Output interface.} 
The output action reads and emits the final active-bit sequence:
\[
  W\big((a'_1,\ldots,a'_n)\big).
\]
The emitted values are determined by the computed residual state. The external environment accepts exactly when every emitted entry is $0$. After output, control returns to the input-waiting process $T$.

\subsubsection{Verification of the Dyck-1-2 Transformer}

We now verify that the Dyck-1-2 Transformer exhibits the expected external behavior: it receives a nonempty input string, computes and emits the corresponding numerical result, and returns to its input-waiting process.
\runhead{State-transition equations.} 
For each $w\in\Sigma^+ $, define
\[
\begin{aligned}
T &= R(w)\circ T_2,\\
T_2 &= \big((\alpha_1\oplus\alpha_2)\circ\alpha_5\big)\circ T_3,\\
T_3 &= \big((\alpha_{12}\oplus\alpha_{18})\circ\alpha_5\big)\circ T_4,\\
T_4 &= F_1\circ T_5,\\
T_5 &= \big((\alpha_{12}\oplus\alpha_{18})\circ\alpha_5\big)\circ T_6,\\
T_6 &= F_2\circ T_7,\\
T_7 &= W\big((a'_1,\ldots,a'_n)\big)\circ T.
\end{aligned}
\]
Here, $F_1$ and $F_2$ are the composite feed-forward terms defined above.

\runhead{Abstraction and external behavior.} 
Use the abstraction set $I$ defined earlier in the manuscript, containing all configured internal actions occurring above. The interface actions $R(w)$ and $W(a'_1,\ldots,a'_n)$ remain visible.

\begin{theorem}[External behavior of the Dyck-1 Transformer.]
For every received input string $w$, the transformer construction of Dyck-1-2 computes and emits final active-bit sequence and returns to
input-waiting state. Thus, $\tau_I(T)$ satisfies the external behavior equation
\[
  \tau_I(T) = R(w)\circ W(a'_1,\ldots,a'_n)\circ\tau_I(T).
\]
\begin{proof}
Applying the abstraction axioms to the state-transition equations and expanding $F_1=U_1\circ C_3\circ C_4$ and $F_2=U_2\circ C_3\circ C_4$ gives
\[
\begin{aligned}
\tau_I(T) &= R(w)\circ\tau_I(T_2),\\
\tau_I(T_2) &= (\tau\oplus\tau)\circ\tau\circ\tau_I(T_3),\\
\tau_I(T_3) &= (\tau\oplus\tau)\circ\tau\circ\tau_I(T_4),\\
\tau_I(T_4) &= \tau_I(U_1)\circ\tau_I(C_3)\circ\tau_I(C_4)\circ\tau_I(T_5),\\
\tau_I(T_5) &= (\tau\oplus\tau)\circ\tau\circ\tau_I(T_6),\\
\tau_I(T_6) &= \tau_I(U_2)\circ\tau_I(C_3)\circ\tau_I(C_4)\circ\tau_I(T_7),\\
\tau_I(T_7) &= W\big((a'_1,\ldots,a'_n)\big)\circ\tau_I(T).
\end{aligned}
\]
Since $U_1$, $U_2$, $C_3$, and $C_4$ each consist of two parallel internal actions followed by addition, TI2, TI6, TI7, and B4 give
\[
\tau_I(U_1)=\tau_I(U_2)=\tau_I(C_3)=\tau_I(C_4)
=(\tau\oplus\tau)\circ\tau
=\tau\circ\tau.
\]
Substituting these equalities, simplifying the remaining parallel silent actions by B4, and applying A5 gives
\[
\begin{aligned}
\tau_I(T) &= R(w)\circ\tau_I(T_2),\\
\tau_I(T_2) &= \tau\circ\tau\circ\tau_I(T_3),\\
\tau_I(T_3) &= \tau\circ\tau\circ\tau_I(T_4),\\
\tau_I(T_4) &= \tau\circ\tau\circ\tau\circ\tau\circ\tau\circ\tau\circ\tau_I(T_5),\\
\tau_I(T_5) &= \tau\circ\tau\circ\tau_I(T_6),\\
\tau_I(T_6) &= \tau\circ\tau\circ\tau\circ\tau\circ\tau\circ\tau\circ\tau_I(T_7),\\
\tau_I(T_7) &= W\big((a'_1,\ldots,a'_n)\big)\circ\tau_I(T).
\end{aligned}
\]
For any continuation $Y$, A5 and B1 establish
\[
R(w)\circ(\tau\circ Y)
\overset{A5}{=}
(R(w)\circ\tau)\circ Y
\overset{B1}{=}
R(w)\circ Y.
\]
Substituting the abstracted continuation equations into each input branch and repeatedly applying this identity removes the finite sequence of internal silent actions. Consequently,
\[
\tau_I(T) = R(w)\circ W\big((a'_1,\ldots,a'_n)\big)\circ\tau_I(T).
\]

The resulting equation exhibits the required external behavior: the transformer receives an input string, emits its computed final active-bit sequence, and returns to its input-waiting process. 
\end{proof}
\end{theorem} 
%\newpage\input{section8.tex}
%\newpage\input{section9.tex}
%\newpage\input{section10.tex}
%\newpage\input{section11.tex}
\bibliographystyle{elsarticle-num}
\newpage\bibliography{Refs-CFMAI}

\end{document}